\documentclass[journal,english]{IEEEtran}
\usepackage{cite}
\usepackage{mdframed}
\usepackage{epsfig,subfigure}
\usepackage{amssymb}
\usepackage[fleqn]{amsmath}
\usepackage{color}
\usepackage[fleqn]{amsmath}
\def\BibTeX{{\rm B\kern-.05em{\sc i\kern-.025em b}\kern-.08em
    T\kern-.1667em\lower.7ex\hbox{E}\kern-.125emX}}
\usepackage{balance}

\usepackage{arydshln}
\def\QED{\mbox{\rule[0pt]{1.5ex}{1.5ex}}}

\usepackage{graphicx}
\usepackage{color}
\usepackage[dvipsnames]{xcolor}

\definecolor{armygreen}{rgb}{0.29, 0.33, 0.13}

\usepackage{amssymb}
\usepackage{amsmath,amsfonts,amssymb}
\usepackage{verbatim}
\usepackage{stfloats}
\usepackage[bookmarks=false]{}
\usepackage{algorithm}
\usepackage{algcompatible}
\usepackage{tikz,pgfplots,tikzpeople}

\newtheorem{theorem}{Theorem}

\newtheorem{definition}{Definition}
\newtheorem{lemma}{Lemma}

\newtheorem{remark}{Remark}

\newtheorem{example}{Example}

\begin{document}
\date{}

\title{On the Extremal Source Key Rates for
\\Secure Storage over Graphs}
\author{Zhou Li
\thanks{Zhou Li is with Guangxi Key Laboratory of Multimedia Communications and Network Technology, Guangxi University, Nanning 530004, China (e-mail: lizhou@gxu.edu.cn).} 
}

\maketitle
\begin{abstract}
This paper investigates secure storage codes over graphs, where multiple independent source symbols are encoded and stored at graph nodes subject to edge-wise correctness and security constraints. For each edge, a specified subset of source symbols must be recoverable from its two incident nodes, while no information about the remaining sources is revealed. To meet the security requirement, a shared source key may be employed. The ratio between the source symbol size and the source key size defines the source key rate, and the supremum of all achievable rates is referred to as the source key capacity.

We study extremal values of the source key capacity in secure storage systems and provide complete graph characterizations for several fundamental settings. For the case where each edge is associated with a single source symbol, we characterize all graphs whose source key capacity equals one. We then generalize this result to the case where each edge is associated with multiple source symbols and identify a broad class of graphs that achieve the corresponding extremal capacity under a mild structural condition. In addition, we characterize all graphs for which secure storage can be achieved without using any source key.

\end{abstract}


\allowdisplaybreaks
\section{Introduction}
\label{sec:intro}

Secure storage systems distribute data across multiple storage nodes while enforcing precise access and information-theoretic security guarantees. Such systems naturally arise in distributed storage networks, cloud platforms, and privacy-preserving data services, where reliability, availability, and confidentiality must be simultaneously ensured. In these settings, storage nodes may be partially observed or compromised, and security must hold even when only a subset of stored data is accessible.

At a fundamental level, a secure storage system stores multiple independent data objects, referred to as source symbols, across a collection of storage nodes. Each node stores a coded symbol that is a function of the source symbols and possibly some shared randomness. Two types of constraints govern the system. Correctness constraints specify which subsets of source symbols must be recoverable from certain collections of storage nodes, while security constraints require that any unauthorized observation reveals no information about protected source symbols. Designing encoding schemes that efficiently satisfy both constraints is a central problem in information-theoretic security.

In this paper, we consider a class of secure storage systems whose data access and secrecy requirements are described by a graph. Specifically, $K$ source symbols $W_1,\ldots,W_K$, each consisting of $L$ bits, together with a shared source key $\mathcal{Z}$ of $L_Z$ bits, are encoded into $N$ stored symbols $V_1,\ldots,V_N$ and placed on the nodes of a graph $G=(\mathcal{V},\mathcal{E})$. Each node corresponds to a storage server, while an edge ${V_i,V_j}\in\mathcal{E}$ models the possibility that an external observer jointly accesses the two associated servers. For notational simplicity, the same symbol $V_n$ is used to represent both a stored symbol and its hosting node. To each edge we assign either a collection of $M$ source symbols or none at all. If an edge is associated with $M$ source symbols, then the pair $(V_i,V_j)$ must enable recovery of those symbols while remaining information-theoretically independent of the other $K-M$ source symbols; if no source symbols are assigned, the pair $(V_i,V_j)$ must reveal no information about any of the $K$ source symbols. This graphical model allows for flexible and heterogeneous access structures across different server pairs. The performance of a secure storage code is quantified by the source key rate $L/L_Z$, and our goal is to determine, for a given graph $G$, the maximum achievable source key rate, referred to as the source key capacity.

\begin{figure}[h]
    \centering
    \includegraphics[width=0.48\textwidth]{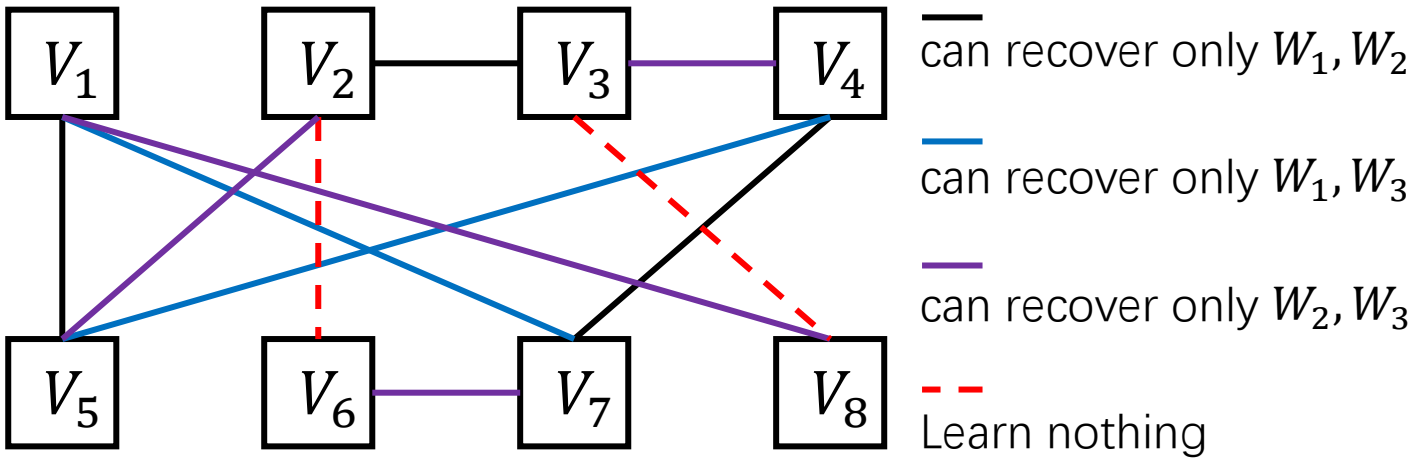}
    \caption{\small
An example of a secure storage problem with $K=3$ source symbols and $N=8$ coded symbols. The source key capacity of this graph is $1/2$ (see Theorem~\ref{thm:d2}; a corresponding code construction is shown in Fig.~\ref{fig5}). This model can be interpreted as storing three files $W_1, W_2, W_3$ over eight servers $V_1,\ldots,V_8$, where certain pairs of servers allow secure retrieval of specific files.
}
    \label{fig0}
\end{figure}

Although presented in a storage-system context, secure storage over graphs is closely related to several communication network problems. When the graph is bipartite, the model generalizes the conditional disclosure of secrets (CDS) problem \cite{SymPIR, Applebaum_Arkis_Raykov_Vasudevan,Li_Sun_CDS, Li_Sun_linearCDS, Wang_Ulukus_CDMS,Li_Zhang_CDSnoise,Li_Zhang_CDSnoisejournal} from a single secret to multiple secrets, allowing arbitrary subsets of source symbols to be conditionally disclosed. Through this connection, secure storage also relates to secret sharing \cite{Beimel_Survey}, particularly for graph-based access structures \cite{Sun_Shieh_Graph}, although existing multi-user secret sharing works \cite{khalesi2021capacity, wu2023capacity} differ in message settings and performance metrics. In addition, secure storage can be viewed as a secure network coding problem \cite{Ahlswede_Cai_etal, secure_nc, Cai_Chan}, and has been studied in storage systems without explicitly considering security over graphs \cite{Li_Sun_storageovergraph}.
Unlike prior work that focuses on multicast or nested message settings \cite{combination_network, Maheshwar_Li_Li, bidokhti2016capacity, Salimi_Liu_Cui}, we study how access and security patterns jointly interact through a general graph structure.

Graph-based formulations have appeared in a variety of network information theory problems. For example, index coding \cite{Yossef_Birk_Jayram_Kol_Trans} and coded caching \cite{Maddah_Ali_Niesen} can both be modeled using graphs or hypergraphs that specify decoding requirements. More recently, \cite{Sahraei_Gastpar} studied a related graph storage problem without security constraints, allowing coded symbols of unequal sizes. These connections suggest that characterizing capacity for general graph-based storage problems is highly nontrivial and often involves long-standing open questions.

To enforce security constraints, the encoding process in a secure storage system may incorporate a shared source key that is independent of all source symbols. The presence of such shared randomness introduces a fundamental tradeoff between security and storage efficiency. We quantify this tradeoff through the \emph{source key rate}, defined as the ratio between the size of each source symbol and the size of the source key. The supremum of all achievable source key rates is referred to as the \emph{source key capacity}. Understanding how the underlying graph structure constrains the source key capacity is the main objective of this work.

Secure storage over graphs has previously been studied in our earlier work \cite{Li_Sun_securestorage}. Using the same graph-based model, that work characterized extremal symbol rates, namely, the maximum achievable storage rate per coded symbol under correctness and security constraints, while allowing an arbitrary shared source key. The present paper, on the other hand, places no restriction on the source symbol rates and focuses on characterizing the extremal behavior of the source key rate. This change in the performance metric leads to different extremal regimes and reveals distinct structural properties of the underlying graphs.

\subsection*{Main Contributions}

We adopt an extremal perspective and study graph structures for which the source key capacity achieves its largest possible values under nontrivial security constraints.

We first consider the case where each edge is associated with a single source symbol. We show that whenever security constraints are present, the source key capacity is at most one, and we provide a complete characterization of all graphs that achieve this extremal value. The characterization is based on an alignment-based framework that captures the interaction among source symbols, security noise, and coded symbols, generalizing prior alignment ideas developed for single-secret CDS problems.

We then extend our analysis to the general setting in which each edge may be associated with multiple source symbols. Under a natural structural condition ensuring that each coded symbol is fully protected by noise, we establish that the source key capacity is upper bounded by the reciprocal of the number of source symbols per edge, and we characterize all graphs achieving this bound. The achievability relies on randomized constructions in higher-dimensional spaces.

Finally, we provide a tight necessary and sufficient condition under which a secure storage code over an arbitrary graph can be constructed without using any source key. Specifically, we show that keyless secure storage is possible if and only if, for every edge, the union of the common source sets of its two incident nodes exactly matches the set of source symbols required by that edge. This result completely characterizes extremal graphs admitting zero source key rate.

\section{Problem Statement and Definitions}\label{sec:model}
Consider an undirected graph $ G = (\mathcal{V}, \mathcal{E}) $ with $N$ nodes. The node set is given by $\mathcal{V} = \{V_1, \ldots, V_N\},$
and the edge set $ \mathcal{E} $ consists of unordered pairs of nodes, i.e., each edge is of the form $ \{V_i, V_j\} \subset \mathcal{V} $.
Each node is regarded as a server that has access to $K$ independent source symbols $\{W_k\}_{k \in [K]}\triangleq\{W_1,\cdots,W_K\}$, where each $W_k$ consists of $L$ independent and identically distributed (i.i.d.) symbols uniformly drawn from a finite field $ \mathbb{F}_q $. In addition, the servers share a common source key $\mathcal{Z}$, which consists of $L_Z$ independent and identically distributed (i.i.d.) symbols uniformly drawn from the same field $ \mathbb{F}_q $.

The source symbols and the key are mutually independent, and their joint entropy satisfies
\begin{eqnarray}
H\left(\{W_k\}_{k\in[K]}, \mathcal{Z}\right)
&=& \sum_{k\in[K]} H\left(W_k\right) + H\left(\mathcal{Z}\right), \label{ind}\\
H(W_k) &=& L,\quad H(\mathcal{Z}) = L_Z.
\label{h1}
\end{eqnarray}

Each node $V_n$, $n \in [N]$, transmits a message, also denoted by $V_n$ for notational simplicity. The message $\mathcal{V}=\{V_n\}_{n\in [N]}$ is a deterministic function of the source symbols and the key, i.e.,
\begin{eqnarray}
    H\!\left(\mathcal{V} \mid \{W_k\}_{k\in[K]}, Z \right) = 0. \label{messagevn}
\end{eqnarray}

Each edge $\{V_i,V_j\} \in \mathcal{E}$ is associated with a subset of source symbols 
$\{W_k\}_{k\in \mathcal{M}}$, where $\mathcal{M} \subseteq [K]$ is either empty or satisfies $|\mathcal{M}|=M$. 
This association can be described by a function $f: \mathcal{E} \to 2^{\{W_1, \dots, W_K\}}$, where $ f(\{V_i, V_j\}) = \{W_k\}_{k \in \mathcal{M}}.$
In other words, for each edge $\{V_i, V_j\}$, only the messages 
$\{W_k\}_{k \in \mathcal{M}}$ can be decoded from this edge. 
If $\mathcal{M}$ is empty, no information is revealed. 
Formally, for all $\{V_i, V_j\} \in \mathcal{E}$ with $f(\{V_i, V_j\}) = \mathcal{M}$, we have
{\setlength{\mathindent}{0cm}
\begin{align}
&(\mbox{Correctness}) ~H\left( \left\{W_k\right\}_{k \in \mathcal{M}} \mid V_i, V_j \right) = 0, \label{dec} \\
&(\mbox{Security}) ~~~~I\left(V_i, V_j; \left\{W_k\right\}_{k \in [K]\backslash\mathcal{M}}| \left\{W_k\right\}_{k \in \mathcal{M}} \right) = 0. \label{sec} 
\end{align}}

To interpret the security requirement, we model the adversary as an external eavesdropper capable of observing any single edge in the network, while having no access to additional information. Nodes with no incident edges impose no security constraints and are therefore trivial. Accordingly, throughout this work we restrict attention to graphs $G$ that contain no isolated nodes.

A \emph{secure storage code} is a mapping from the source symbols $\{W_k\}_{k\in[K]}$ and the source key $\mathcal{Z}$ to the coded symbols $\{V_n\}_{n\in[N]}$ that satisfies the correctness constraint~(\ref{dec}) and the security constraint~(\ref{sec}) specified by a graph $G=(\mathcal{V},\mathcal{E})$.

The achievable source key rate of a secure storage code is defined as
\begin{eqnarray}
R_Z \triangleq \frac{L}{L_Z}, \label{rate}
\end{eqnarray}
and the \emph{source key capacity}, denoted by $C$, is the supremum of all achievable rates. In particular, this supremum is taken in the asymptotic sense, so that rates of the form $R=\lim_{L\to\infty} L/L_Z$ are included.

\subsection{Graph Definitions}
To facilitate the presentation of our results, we introduce several graph-theoretic definitions in this subsection.

Given a graph $G=(\mathcal{V},\mathcal{E})$, we focus on a single source symbol $W_k$ and examine which edges are associated with it. This motivates the definition of the characteristic graph $G^{[k]}$.
\begin{definition}[Characteristic Graph $G^{[k]}$ of $W_k$] \label{def:cha}
Let $G=(\mathcal{V},\mathcal{E})$ be a graph. For each $k\in[K]$, define the characteristic graph $G^{[k]} = (\mathcal{V}^{[k]}, \mathcal{E}^{[k]}),$
where
\begin{eqnarray}
    \mathcal{V}^{[k]} &\triangleq& \{V^{[k]}_1, \ldots, V^{[k]}_N\}, \\
\{V^{[k]}_i, V^{[k]}_j\} \in \mathcal{E}^{[k]}
&\iff&
\{V_i, V_j\} \in \mathcal{E}. 
\end{eqnarray}
For each edge $\{V^{[k]}_i, V^{[k]}_j\}\in \mathcal{E}^{[k]}$, define the associated label function
\begin{eqnarray}
    f^{[k]}(\{V^{[k]}_i, V^{[k]}_j\}) \triangleq
\begin{cases}
\{W_k\}, & \text{if } W_k \in f(\{V_i, V_j\}),\\
\emptyset, & \text{otherwise}.
\end{cases}
\end{eqnarray}
\end{definition}

Fig.~\ref{fig1} illustrates an example of the graph $G$ and its corresponding characteristic graph $G^{[1]}$ associated with $W_1$.

\begin{figure}[h]
    \centering
    \includegraphics[width=0.48\textwidth]{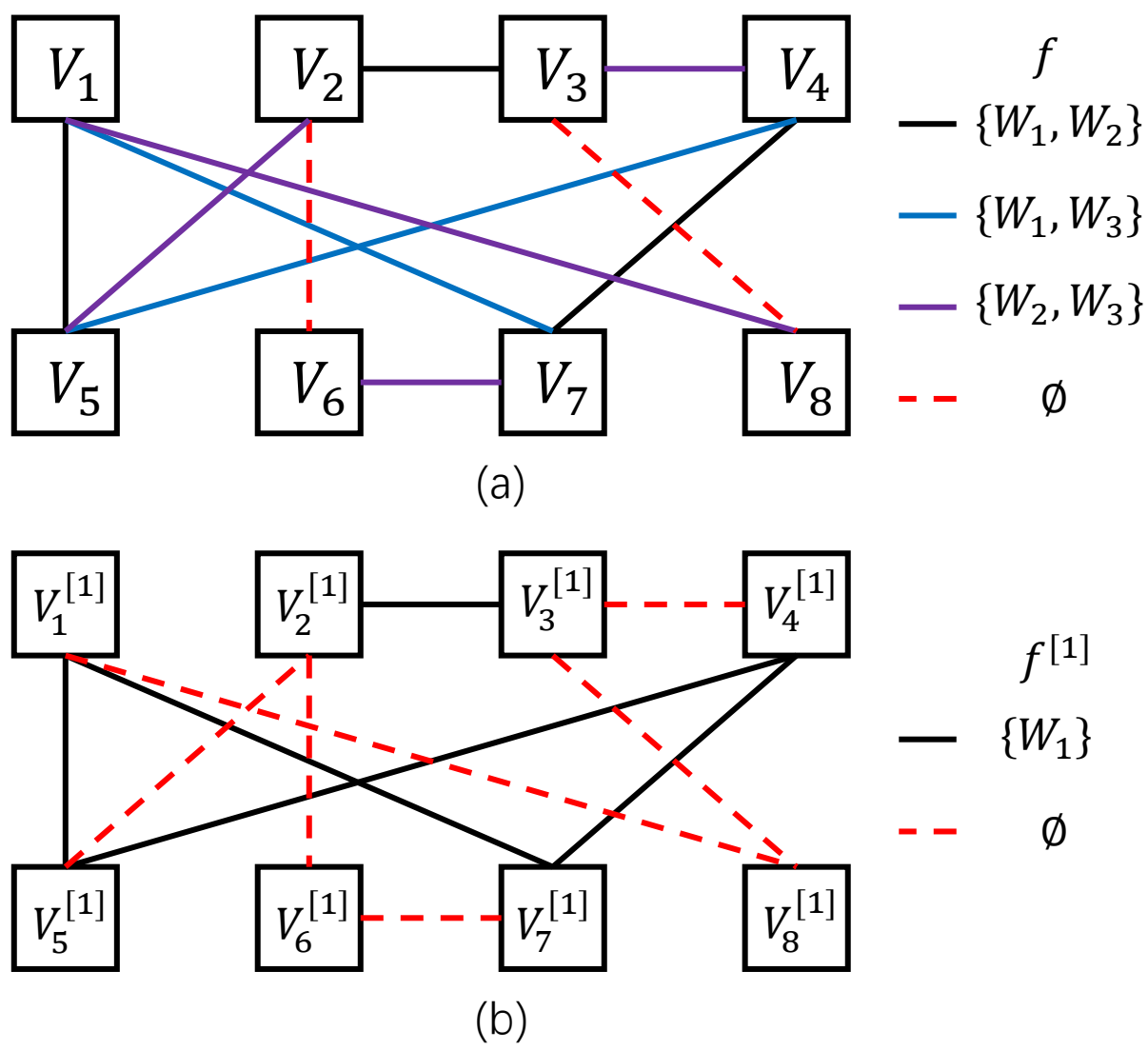}
    \caption{\small (a) An example graph $G$ with $K=3$ source symbols, $N=8$ coded symbols, and $M=2$, where each edge is associated with a subset of two source symbols; (b) the corresponding characteristic graph $G^{[1]}$ for source symbol $W_1$.}
    \label{fig1}
\end{figure}

In the characteristic graph $G^{[k]}=(\mathcal{V}^{[k]},\mathcal{E}^{[k]})$, a key distinction lies in whether an edge carries information about any source symbol. Accordingly, edges are classified as either qualified or unqualified, and this classification naturally extends to paths and connected components.

\begin{definition}[Qualified and Unqualified Edges, Paths, and Components]
Consider a graph $G = (\mathcal{V}, \mathcal{E})$.  
An edge $E \in \mathcal{E}$ is called \emph{qualified} if $f(E) \neq \emptyset$, and \emph{unqualified} if $f(E) = \emptyset$.  

A \emph{qualified} (resp. \emph{unqualified}) path is a path consisting entirely of qualified (resp. unqualified) edges. A qualified edge is called \emph{internal} if its two endpoints are connected by an unqualified path.

A \emph{qualified} (resp. \emph{unqualified}) component is a maximal induced subgraph of $G$ in which every pair of nodes is connected by a qualified (resp. unqualified) path.
\end{definition}

Note that the above definition applies to both $G$ and $G^{[k]}$.

As illustrated in Fig.~\ref{fig1}, $\{V_1, V_5\}$ is a qualified edge, while $\{V_1^{[1]}, V_8^{[1]}\}$ is an unqualified edge.  
The sequence $(\{V_1^{[1]}, V_8^{[1]}\}, \{V_8^{[1]}, V_3^{[1]}\}, \{V_3^{[1]}, V_4^{[1]}\})$ forms an unqualified path.  
Furthermore, $G$ constitutes a qualified component, and $G^{[1]}$ contains no internal qualified edges.  

In Fig.~\ref{fig4}, $\{V_1^{[1]}, V_3^{[1]}\}$ is an internal qualified edge.

In a graph $G=(\mathcal{V},\mathcal{E})$, our analysis requires identifying the source symbols that are shared by all edges incident to a given node. To formalize this notion, we introduce the following definition.

\begin{definition}[Common Sources $\mathcal{C}(V)$]
Let $V \in \mathcal{V}$ be a node in a graph $G = (\mathcal{V}, \mathcal{E})$.  
The set of common source symbols at node $V$ is defined as
\begin{eqnarray}
\mathcal{C}(V) \triangleq \bigcap_{i:\,\{V,V_i\}\in\mathcal{E}} f(\{V,V_i\}).
\end{eqnarray}
\end{definition}

For example, consider node $V_1$ in Fig.~\ref{fig1}:  $\mathcal{C}(V_1) = \{W_1, W_2\} \cap \{W_2, W_3\} \cap \{W_1, W_3\} = \emptyset.$
Similarly, for node $V_6$, we have $\mathcal{C}(V_6) = \{W_2, W_3\} \cap \emptyset = \emptyset.$  

For node $V_3$ in Fig.~\ref{fig6},  $\mathcal{C}(V_3) = \{W_1, W_2\} \cap \{W_1, W_3\} = \{W_1\}.$




\vspace{0.05in}

Finally, we identify a class of nodes whose incident edges impose no conflicting constraints. Specifically, a node $V$ in a graph $G=(\mathcal{V},\mathcal{E})$ is said to be \emph{degenerate} if all edges incident to $V$ are associated with the same set of source symbols. In this case, all constraints at $V$ can be satisfied by storing an identical set of source symbols at that node.

For notational convenience, degenerate nodes can be removed without affecting the validity of our results. This removal serves only to simplify the presentation; all statements and proofs remain valid even in the presence of degenerate nodes. We formalize this reduction as follows.

\begin{definition}[Non-degenerate Subgraph $\widetilde{G}$ of $G$]
Let $G=(\mathcal{V},\mathcal{E})$ be a graph. Define the set of degenerate nodes as
\begin{eqnarray}
\mathcal{V}_d \triangleq \bigl\{ V\in\mathcal{V} \,\big|\, 
f(\{V,V_i\})=\mathcal{C}(V),\ \forall\,\{V,V_i\}\in\mathcal{E} \bigr\}.
\end{eqnarray}
The \emph{non-degenerate subgraph} of $G$ is the induced subgraph on the node set $\mathcal{V}\setminus\mathcal{V}_d$, denoted by
 $\widetilde{G} \triangleq G[\mathcal{V}\backslash\mathcal{V}_d]$.
\end{definition}

\section{Results}
In this section, we present our main results, together with illustrative examples and related observations.

\subsection{$M=1$ and Extremal Graphs with $C=1$}
We first focus on the case $M=1$, for which extremal graphs with source key capacity $C=1$ admit a complete characterization.

\begin{theorem}\label{thm:d1}
For $M=1$, the source key capacity of a secure storage code over a graph $G$ equals $C=1$ if and only if the non-degenerate subgraph $\widetilde{G}$ of $G$ is nonempty and, for every $k\in[K]$, the characteristic graph $G^{[k]}$ contains no internal qualified edge.
\end{theorem}

\begin{remark}
When $\widetilde{G}$ is empty, all nodes in $G$ are incident only to edges associated with the same set of source symbols. It follows that if $G$ contains only a single type of qualified edge, no source key is needed to ensure security (see Fig.~\ref{fig2} for an example).
\end{remark}

\begin{figure}[h]
    \centering
    \includegraphics[width=0.48\textwidth]{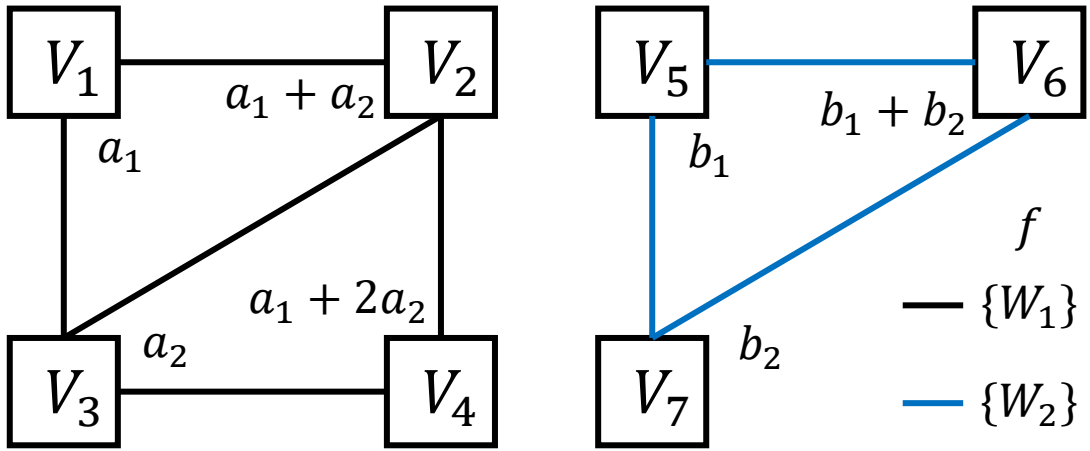}
    \caption{\small  An example of a graph $G$ with $K=2$, $N=7$, and $M=1$, where $\widetilde{G}=\emptyset$. Here, $W_1=(a_1,a_2)$ and $W_2=(b_1,b_2)$, and no key is required.
}
    \label{fig2}
\end{figure}

The proof of Theorem~\ref{thm:d1} appears in Section~\ref{sec:thm1}. To build intuition, we consider two examples. The first example in Fig.~\ref{fig3} illustrates the sufficiency of the conditions in Theorem~\ref{thm:d1}, under which the source key capacity equals $C=1$.

\begin{figure}[h]
    \centering
    \includegraphics[width=0.48\textwidth]{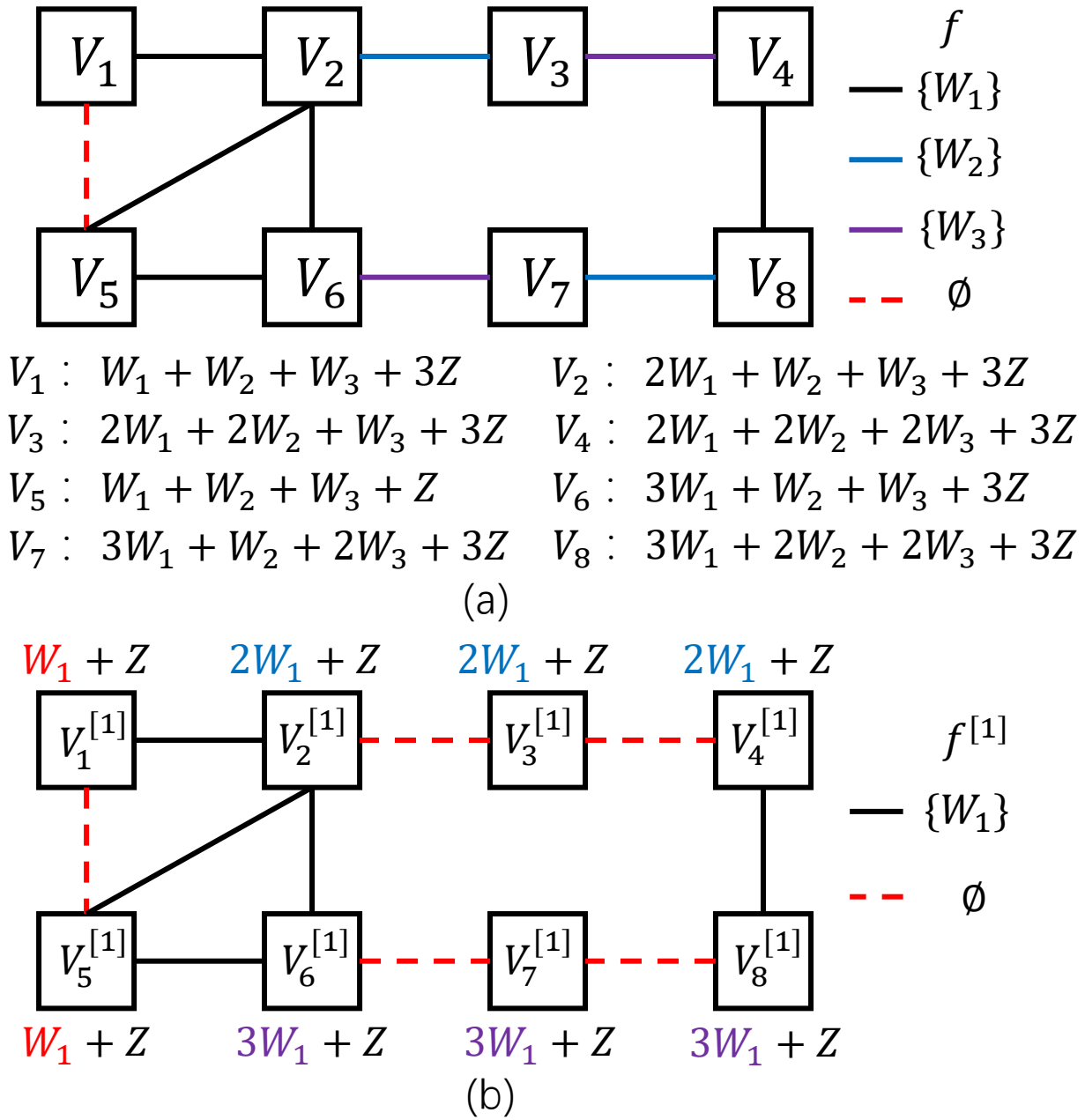}
    \caption{\small
(a) An example graph $G$ with $K=3$, $N=8$, and $M=1$, and (b) the characteristic graph $G^{[1]}$ corresponding to $W_1$. The source key capacity over $G$ is $1$, and a capacity-achieving code construction is illustrated.
}
    \label{fig3}
\end{figure}

\begin{example}\label{ex1}
Consider the secure storage instance shown in Fig.~\ref{fig3}. Each node is subject to the security requirement $L_Z \geq L$ and the key rate constraint $R_Z \leq 1$. We show that the capacity $R_Z=1$ is achievable by explicitly constructing a coding scheme.

Assume that each source symbol $W_k$ takes values in $\mathbb{F}_5$. We first introduce a single random variable $Z$, uniformly distributed over $\mathbb{F}_5$, which is shared across all nodes. This common noise is used at every node to satisfy the security constraints, a strategy referred to as \emph{noise alignment} (see Lemma~\ref{lemma:noise}).

Next, the coded symbols are designed by treating each source symbol $W_k$ independently through its characteristic graph $G^{[k]}$. Consider $W_1$ and the corresponding characteristic graph $G^{[1]}$ depicted in Fig.~\ref{fig3}(b). Since an unqualified component cannot disclose any information about $W_1$, all nodes within the same unqualified component are assigned identical coded symbols, a procedure known as \emph{coded symbol alignment} (see Lemma~\ref{lemma:signal}). Specifically, we set $V_1^{[1]}=V_5^{[1]}=W_1+Z,\quad
V_2^{[1]}=V_3^{[1]}=V_4^{[1]}=2W_1+Z,\quad
V_6^{[1]}=V_7^{[1]}=V_8^{[1]}=3W_1+Z,$
where different unqualified components are assigned distinct linear combinations, as indicated by the coloring in Fig.~\ref{fig3}(b).

Because $G^{[1]}$ contains no internal qualified edge, every qualified edge connects two different unqualified components. Consequently, the coded symbols observed on such an edge form linearly independent combinations of $W_1$ and $Z$, allowing $W_1$ to be recovered, as illustrated by the edge $\{V_1^{[1]},V_5^{[1]}\}$ in Fig.~\ref{fig3}(b).

Finally, the overall coded symbol at each node is obtained by summing the contributions from all $G^{[k]}$. This construction guarantees that, for every edge, the desired source symbol appears with distinct coefficients to ensure correctness, while all undesired source symbols and the noise are aligned, thereby ensuring security.
\end{example}


The second example (see Fig.~\ref{fig4}) is used to explain the `only if' part, i.e., when $\widetilde{G} \neq \emptyset$ and the graph $G$ does not satisfy the condition in Theorem \ref{thm:d1}, 
then the source key capacity $C < 1$.

\begin{figure}[h]
    \centering
    \includegraphics[width=0.48\textwidth]{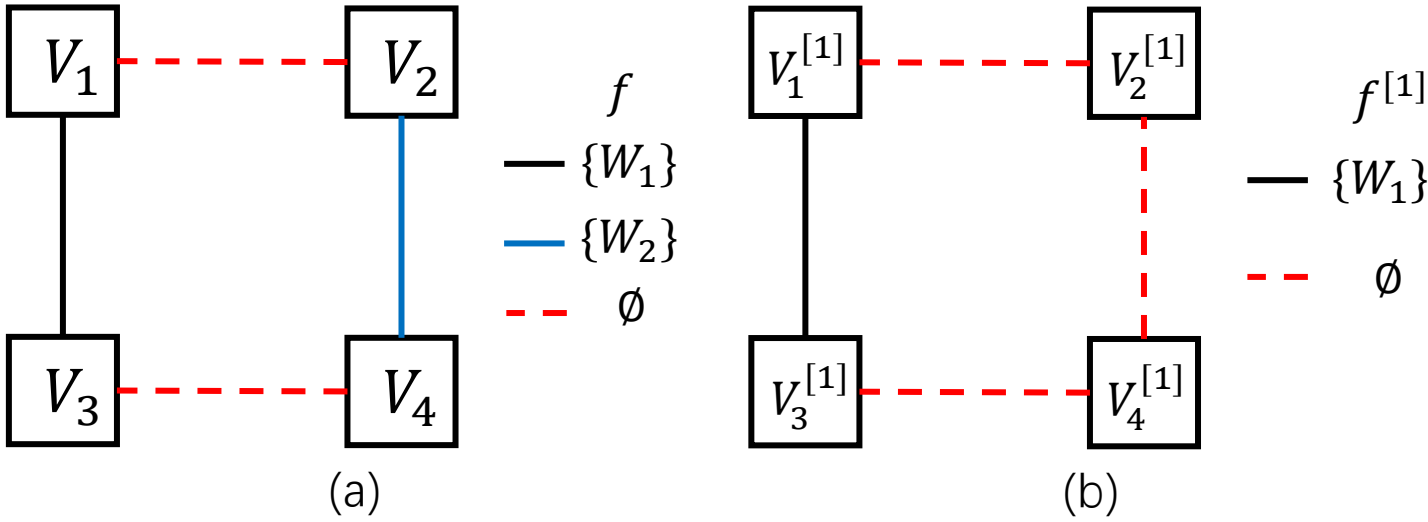}
    \caption{\small   (a) An example graph $G$ with $K=2, N=4, M=1$ and (b) its $G^{[1]}$ of $W_1$. The source key capacity over $G$ cannot be $1$.
}
    \label{fig4}
\end{figure}

\begin{example}\label{ex2}
Consider the secure storage instance illustrated in Fig.~\ref{fig4}. In this example, the graph $G^{[1]}$ contains a qualified edge $\{V_1^{[1]},V_3^{[1]}\}$ that lies entirely within an unqualified path composed of the edges $\{V_1^{[1]},V_2^{[1]}\}$, $\{V_2^{[1]},V_4^{[1]}\}$, and $\{V_4^{[1]},V_3^{[1]}\}$. We argue that the source key capacity cannot be $C=1$, or equivalently $L_Z \neq L$, by providing an intuitive contradiction-based explanation (neglecting $o(L)$ terms).

Suppose, for the sake of argument, that $L_Z=L$. Then, as shown in Fig.~\ref{fig4}(a), all nodes $V_1,V_2,V_3,$ and $V_4$ must employ a common noise variable in order to satisfy the security constraints, following the noise alignment principle in Lemma~\ref{lemma:noise}. Meanwhile, Fig.~\ref{fig4}(b) shows that the path $\left(\{V_1^{[1]},V_2^{[1]}\},\{V_2^{[1]},V_4^{[1]}\},\{V_4^{[1]},V_3^{[1]}\}\right)$ forms an unqualified path. As a result, all nodes along this path are forced to store identical coded symbols related to $W_1$, a phenomenon referred to as coded symbol alignment and formally characterized using conditional entropy arguments (see Lemma~\ref{lemma:signal}).

Consequently, nodes $V_1$ and $V_3$ must contain the same information about $W_1$. This, however, contradicts the requirement that $W_1$ be recoverable from the qualified edge $\{V_1,V_3\}$. Therefore, the assumption $L_Z=L$ cannot hold in this setting. A rigorous proof of this argument, based on entropy inequalities, is provided in Section~\ref{sec:thm21}.
\end{example}

\subsection{Arbitrary $M$ and Extremal Graphs with $C=1/M$}

We now generalize Theorem~\ref{thm:d1} to the case of an arbitrary number of associated source symbols per edge, i.e., $M$. We focus on a class of graphs in which the non-degenerate subgraph is nonempty and each non-degenerate node is not associated with any source symbol. For this class, all extremal graphs achieving a source key capacity of $C=1/M$ can be fully characterized by the following result.

\begin{theorem}\label{thm:d2}
Let $G=(\mathcal{V},\mathcal{E})$ be a graph such that its non-degenerate subgraph $\widetilde{G}$ is nonempty and satisfies $\mathcal{C}(V)=\emptyset$ for all $V\in\mathcal{V}\backslash\mathcal{V}_d$.
Then, for this class of graphs, the source key capacity of a secure storage code equals $C=1/M$ if and only if, for every $k\in[K]$, the characteristic graph $G^{[k]}$ of the source symbol $W_k$ contains no internal qualified edge.
\end{theorem}

\begin{remark}
Theorem~\ref{thm:d2} subsumes Theorem~\ref{thm:d1} as a special case. Indeed, when $M=1$, every non-degenerate node necessarily satisfies $\mathcal{C}(V)=\emptyset$, since by definition a non-degenerate node must be incident to edges associated with distinct source symbols.
\end{remark}

The proof of Theorem \ref{thm:d2} is presented in Section \ref{sec:thm2}. An example is given  in Fig.~\ref{fig5} to explain the code construction of the `if' part.

\begin{figure}[h]
    \centering
    \includegraphics[width=0.48\textwidth]{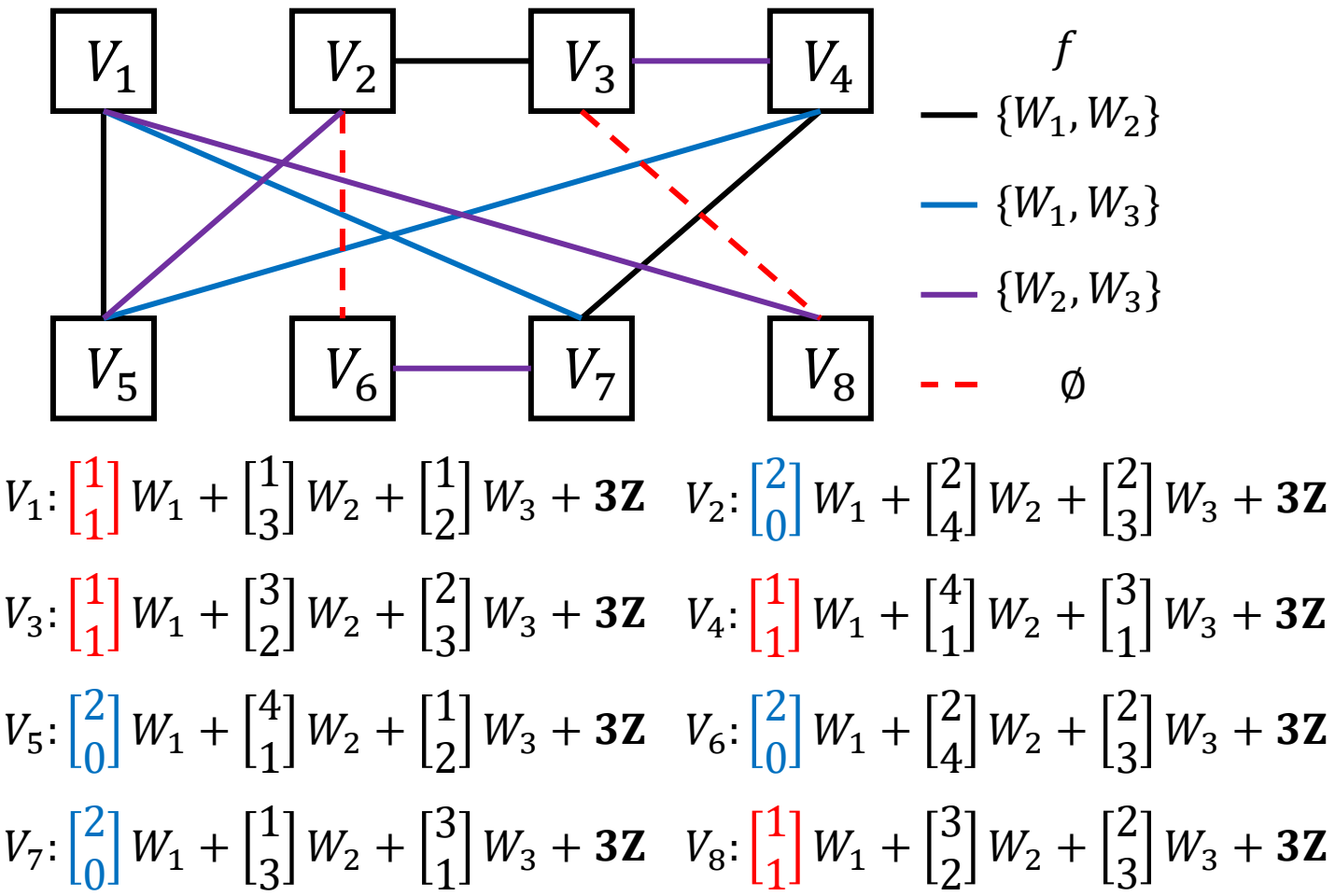}
    \caption{\small
An example graph $G$ with $K=3$, $N=8$, and $M=2$, together with a code construction that achieves a source key capacity of $1/2$. The precoding coefficients for each source symbol are vectors, whereas only scalar coefficients are required in \cite{Li_Zhang_CDSnoisejournal}.
}
    \label{fig5}
\end{figure}

\begin{example}
Consider the graph $G$ shown in Fig.~\ref{fig5}. The coding strategy follows the same high-level principle as in Example~\ref{ex1}. Specifically, each source symbol $W_k$ is treated independently together with its characteristic graph $G^{[k]}$. For each unqualified component of $G^{[k]}$, a generic linear combination is assigned. For $W_1$, these components are indicated by different colors in Fig.~\ref{fig5}. The final coded symbol stored at each node is obtained by summing the contributions from all graphs $G^{[k]}$.

In this example, each $W_k$ takes values in $\mathbb{F}_5$, and an independent noise vector ${\bf Z}\in\mathbb{F}_5^{2\times 1}$ is employed, where ${\bf Z}=(Z_1,Z_2)^{\mathsf T}$ is uniformly distributed. Compared with Example~\ref{ex1}, which corresponds to the case $M=1$, the case $M=2$ imposes a more stringent correctness requirement. Instead of merely ensuring distinct coefficients for the desired source symbol, the resulting coefficient matrix must be full rank.

Designing such coefficients is generally nontrivial. For this small instance, an explicit construction is illustrated in Fig.~\ref{fig5}. For the general case, the achievability proof in Section~\ref{sec:thm22} relies on a randomized coding argument.
\end{example}

\subsection{Arbitrary $M$ and Extremal Graphs without any Source Key}

Finally, we consider extremal graphs for which a secure storage code can be constructed without using any source key. All such graphs are completely characterized by the following theorem.

\begin{theorem}\label{thm:2d}
A secure storage code over a graph $G = (\mathcal{V}, \mathcal{E})$ can be constructed without using any source key, i.e., $L_Z = 0$, if and only if, for every edge $\{V_i, V_j\} \in \mathcal{E}$,
$\mathcal{C}(V_i) \cup \mathcal{C}(V_j) = f(\{V_i, V_j\})$.
\end{theorem}

In words, Theorem~\ref{thm:2d} requires that, for each edge, the union of the common source sets associated with its two incident nodes coincides with the set of $M$ desired source symbols. The intuition behind this condition is as follows. The total storage available on any edge is exactly $M \times L$, ignoring $o(L)$ terms, and must be entirely occupied by the $M$ desired source symbols, leaving no room for any additional information. Note that an edge may be either qualified or unqualified. In the latter case, $M=0$.

\begin{figure}[h]
    \centering
    \includegraphics[width=0.48\textwidth]{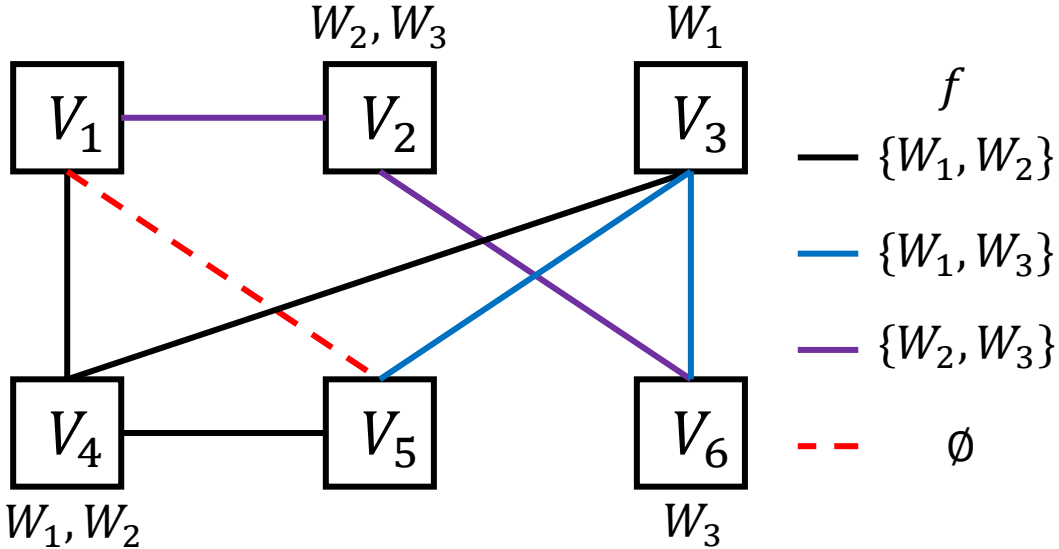}
    \caption{\small
An example graph $G$ with $K=3$, $N=6$, and $M=2$, together with a secure storage code construction without using any source key ($L_Z = 0$).
}
    \label{fig6}
\end{figure}

We first establish necessity. As a direct consequence of the storage constraint, each coded symbol must be a deterministic function of its common sources, as formalized in Lemma~\ref{lemma:det}. If the condition in Theorem~\ref{thm:2d} were violated, the desired source symbols alone would not provide sufficient information to fully populate the storage of an edge, making it impossible to satisfy the decoding requirement without a source key.

We next argue sufficiency. When the condition in Theorem~\ref{thm:2d} holds, random linear coding guarantees that storing a sufficient number of generic linear combinations of the common sources enables reliable recovery of the desired source symbols. An illustrative example is provided in Fig.~\ref{fig6}. The complete proof of Theorem~\ref{thm:2d} is deferred to Section~\ref{sec:thm3}.

\section{Proof of Theorem~\ref{thm:d1}}\label{sec:thm1}
Theorem~\ref{thm:d1} can be viewed as a special case of Theorem~\ref{thm:d2}. 
Consequently, its proof follows from that of Theorem~\ref{thm:d2}, which is provided in Section~\ref{sec:thm2}. 
Nevertheless, we include a separate proof here to explicitly present the code construction for the \emph{if} part. 
This is possible because when $M=1$, an explicit construction can be given, whereas for the general case of Theorem~\ref{thm:d2} with arbitrary $M$, the proof relies on a randomized construction.

\subsection{Proof of the ``If'' Part}\label{sec:thm11}
We show that if the graph $G$ satisfies the conditions stated in Theorem~\ref{thm:d1}, then a secure storage code achieving source key rate $R_Z=1$ exists. 
Assume that each source symbol $W_k$ has length $L$ bits. 
We construct the code over a finite field $\mathbb{F}_q$, where each source symbol is represented by a single field element; the field size $q$ will be specified later in the proof. 
The source key is also chosen as one symbol from $\mathbb{F}_q$, and therefore has length $L_Z=L$ bits. 
It follows that the achieved source key rate is $R_Z=L/L_Z=1$.

A degenerate node $V\in\mathcal{V}_d$, for which all incident edges are associated with the same source symbol $W_i$, is trivial, since its constraints can be satisfied by simply storing $W_i$ at $V$. 
If all incident edges of $V$ are unqualified, we instead assign an independent noise variable to $V$. 
Therefore, without loss of generality, we restrict our attention to the non-degenerate subgraph $\widetilde{G}$ induced by the node set $\mathcal{V}\setminus\mathcal{V}_d$.

For each source symbol $W_k$, $k\in[K]$, we construct the coded symbols by considering its characteristic graph $G^{[k]}$ separately. 
The final coded symbol assignment for $G$ is then obtained by combining the assignments corresponding to all $k\in[K]$.

We construct a secure storage code by encoding each source separately and then aggregating the encodings across all sources.

\textbf{Encoding on $G^{[k]}$:}
Fix an index $k\in[K]$ and consider the corresponding characteristic graph $G^{[k]}$, which consists of $U^{[k]}$ unqualified components. Choose a prime field size $q$ satisfying $q > \max_{k\in[K]} U^{[k]}$.

For each unqualified component indexed by $u\in[U^{[k]}]$, all nodes within this component are assigned an identical coded symbol. Specifically, for any node $V^{[k]}$ belonging to the $u$-th unqualified component of $G^{[k]}$, we define
\begin{eqnarray}
V^{[k]} = u W_k + Z, \label{eq:c1_new}
\end{eqnarray}
where $Z$ denotes an i.i.d. uniform random variable over $\mathbb{F}_q$, independent of all source symbols. The same source key $Z$ is reused for all $k\in[K]$.

Since the condition of Theorem~\ref{thm:d1} guarantees that $G^{[k]}$ contains no internal qualified edge, the assignment in \eqref{eq:c1_new} has the following properties.
\begin{itemize}
\item For any qualified edge $\{V_i^{[k]}, V_j^{[k]}\}$ in $G^{[k]}$, the two nodes belong to different unqualified components and thus have distinct coefficients of $W_k$. Hence,
\begin{eqnarray}
V_i^{[k]} - V_j^{[k]} = (u_i - u_j) W_k,
\label{eq:p1_new}
\end{eqnarray}
from which $W_k$ can be recovered.
\item For any unqualified edge $\{V_i^{[k]}, V_j^{[k]}\}$ in $G^{[k]}$, both nodes lie in the same unqualified component and therefore
\begin{eqnarray}
V_i^{[k]} = V_j^{[k]}.
\label{eq:p2_new}
\end{eqnarray}
\end{itemize}

{\bf Aggregation over $G$:}
There is a one-to-one correspondence between nodes in $G$ and nodes in each $G^{[k]}$. For each node $V$ in $G$, define
\begin{eqnarray}
V = \sum_{k\in[K]} V^{[k]}.
\label{eq:c2_new}
\end{eqnarray}

The source key rate is given by $R_Z=L/L_Z=1$, as desired.



{\bf Poof of Correctness:}
Edges incident to degenerate nodes are trivial, so we only consider non-degenerate edges. Let $\{V_i, V_j\}$ be a qualified edge in $G$ with $f(\{V_i,V_j\}) = \{W_l\}$. Then
\begin{eqnarray}
    &&V_i - V_j \notag\\
&\overset{(\ref{eq:c2_new})}{=}& \sum_{k\in[K]} \bigl( V_i^{[k]} - V_j^{[k]} \bigr) \\
&\overset{(\ref{eq:c1_new})(\ref{eq:p2_new})}{=}& V_i^{[l]} - V_j^{[l]}\\
&\overset{(\ref{eq:c1_new})(\ref{eq:p1_new})}{=}&(u_i-u_j)W_l,
\end{eqnarray}
The second equality follows from the fact that, for $k = l$, the edge $\{V_i^{[k]}, V_j^{[k]}\}$ is a qualified edge and hence contributes $V_i^{[l]} - V_j^{[l]}$ according to \eqref{eq:p1_new}, whereas for all $k \in [K]\setminus\{l\}$, the corresponding edges are unqualified and thus vanish by \eqref{eq:p2_new}.

Since $\{V_i^{[l]}, V_j^{[l]}\}$ is a qualified edge in $G^{[l]}$, the source symbol $W_l$ can be uniquely recovered from $V_i^{[l]} - V_j^{[l]}$ using \eqref{eq:p1_new}. Therefore, the correctness condition is satisfied.

{\bf Proof of Security:}
The pair $(V_i, V_j)$ is equivalent to $(V_i - V_j, V_i)$. As shown above, $V_i - V_j$ reveals exactly $W_l$. Moreover, from \eqref{eq:c1_new} and \eqref{eq:c2_new}, the symbol $V_i$ is masked by independent uniform noise variables $\{Z_k\}_{k\in[K]}$ and hence reveals no information about any source symbols other than $W_l$. Therefore, the security constraint holds.

For any unqualified edge $\{V_i, V_j\}$, we have $V_i^{[k]} = V_j^{[k]}$ for all $k$, implying $V_i = V_j$. Since this common value is independent of all source symbols, security is trivially satisfied.

This completes the proof.

\section{Proof of Theorem \ref{thm:d2}} \label{sec:thm2}
This section provides the proof of Theorem~\ref{thm:d2}. We first prove the “only if” part in Section~\ref{sec:thm21}, and then prove the “if” part in Section~\ref{sec:thm22}.

\subsection{Proof of the ``Only If'' Part}\label{sec:thm21}

We begin with a useful lemma that holds for any source key rate and any graph. This lemma will also be used in the proof of Theorem~\ref{thm:2d}.

\begin{lemma}[Independence of Non-Common Sources]\label{lemma:ind}
A coded symbol $V$ is independent of its non-common source symbols, both with and without conditioning on the common source symbols, i.e.,
\begin{eqnarray}
I\!\left(V; \{ W_k \}_{k\in [K]\setminus \mathcal{C}(V)} \mid \{W_k\}_{k\in\mathcal{C}(V)} \right) &=& 0, \label{eq:ind}\\
I\!\left(V; \{ W_k \}_{k\in [K]\setminus \mathcal{C}(V)} \right) &=& 0. \label{eq:ind0}
\end{eqnarray}
\end{lemma}

{\it Proof:}
We first prove~\eqref{eq:ind}. Consider any non-common source symbol $W_i$ of node $V$, that is, $i \in [K]\setminus \mathcal{C}(V)$, 
Since $W_i$ is not a common source symbol of $V$, there exists an edge $\{V, V_j\}$ such that
$W_i \notin f(\{V, V_j\}) = \{W_k\}_{k\in \mathcal{M}}$.
By the security constraint~\eqref{sec}, we have
\begin{align}
0 \overset{(\ref{sec})}{=}&
I\!\left(V, V_j; \{W_k\}_{k \in [K]\setminus\mathcal{M}} \mid \{W_k\}_{k \in \mathcal{M}} \right) \notag\\
\geq&
I\!\left(V; W_i \mid \{W_k\}_{k \in [K]\setminus\{i\}} \right) \label{eq:ind1}\\
=&
I\!\left(V, \{W_k\}_{k \in [K]\setminus(\{i\}\cup \mathcal{C}(V))}; W_i
\mid \{W_k\}_{k \in \mathcal{C}(V)} \right) \notag\\
&-\underbrace{I\!\left(\{W_k\}_{k \in [K]\setminus(\{i\}\cup \mathcal{C}(V))}; W_i
\mid  \{W_k\}_{k \in \mathcal{C}(V)} \right)}_{\overset{(\ref{ind})}{=}0} \label{eq:ind2}\\
\geq&
I\!\left(V; W_i \mid  \{W_k\}_{k \in \mathcal{C}(V)} \right).
\label{eq:ind3}
\end{align}
The second term in~\eqref{eq:ind2} is zero because the source symbols $\{W_k\}_{k\in[K]}$ are mutually independent.

We are now ready to establish~\eqref{eq:ind}:
\begin{align}
&I\!\left(V; \{ W_k \}_{k\in [K]\setminus \mathcal{C}(V)} \mid \{W_k\}_{k\in\mathcal{C}(V)} \right)\notag\\
=&
\sum_{i\in [K]\setminus \mathcal{C}(V)}
I\!\left(V; W_i \mid \{ W_j \}_{j\in [K]\setminus \mathcal{C}(V),j<i},\{W_k\}_{k\in\mathcal{C}(V)} \right) \notag\\
\leq&
\sum_{i\in [K]\setminus \mathcal{C}(V)}
I\!\left(V; W_i \mid \{W_k\}_{k\in\mathcal{C}(V)}\right)\\
\overset{(\ref{eq:ind3})}{=}& 0.
\end{align}

Next, we prove~\eqref{eq:ind0} as a direct consequence of~\eqref{eq:ind}. Specifically,
\begin{eqnarray}
0 &\overset{(\ref{eq:ind})}{=}&
I\!\left(V; \{ W_k \}_{k\in [K]\setminus \mathcal{C}(V)} \mid \{W_k\}_{k\in\mathcal{C}(V)} \right) \notag\\
&=&
I\!\left(V, \{W_k\}_{k\in\mathcal{C}(V)}; \{ W_k \}_{k\in [K]\setminus \mathcal{C}(V)} \right)\notag\\
&& -\underbrace{I\left( \left\{W_k\right\}_{k\in\mathcal{C}(V)}; \left\{ W_k \right\}_{k\in [K]\backslash \mathcal{C}(V)}  \right)}_{\overset{(\ref{ind})}{=}0}\label{eq:ind4}\\
&\geq&
I\!\left(V; \{ W_k \}_{k\in [K]\setminus \mathcal{C}(V)} \right),
\end{eqnarray}
where the second term of (\ref{eq:ind4}) is zero due to the mutual independence of the source symbols $\{W_k\}_{k\in [K]}$.
\hfill$\QED$

We now proceed to the proof of the ``only if'' part.
We show that the source key rate $R_Z = 1/M$ is not achievable if a graph $G$ does not satisfy the condition in Theorem~\ref{thm:d2}.
Specifically, suppose there exists a non-degenerate subgraph $\widetilde{G}$ of $G$ such that the characteristic graph $G^{[k]}$ of some source symbol $W_k$ contains an internal qualified edge.

Without loss of generality, let $k=1$, and suppose that the internal qualified edge is $\{V_1^{[1]}, V_P^{[1]}\}$, which lies within the sequence of unqualified edges $\bigl(\{V_1^{[1]}, V_2^{[1]}\}, \ldots, \{V_{P-1}^{[1]}, V_P^{[1]}\}\bigr).$
We establish the converse via contradiction. Assume that the rate $R_Z = \frac{1}{M} = \lim_{L\to\infty} \frac{L}{L_Z}$ is asymptotically achievable, i.e.,
$L_Z = ML + o(L)$.
Note that if the rate is exactly achievable, the $o(L)$ term vanishes, and the following argument still applies.

We first show that when $R_Z = 1/M$, every non-degenerate coded symbol with no common source symbols must be fully masked by noise.
This key property is formalized in the following lemma.

\begin{lemma}[Message and Noise Size]\label{lemma:size}
When $R_Z = 1/M$, i.e., $L_Z = ML + o(L)$, for a non-degenerate graph
$\widetilde{G} = (\mathcal{V}, \mathcal{E})$ such that every node $V \in \mathcal{V}$ satisfies $\mathcal{C}(V) = \emptyset$, we have
\begin{eqnarray}
H(V)
&=& H\!\left(V \mid \{W_k\}_{k \in [K]\setminus\{1\}}\right)
= H\!\left(V \mid \{W_k\}_{k \in [K]}\right) \notag\\
&=& H(\mathcal{Z})
= ML + o(L).
\label{eq:size}
\end{eqnarray}
\end{lemma}

{\it Proof:}
Since $\mathcal{C}(V)=\emptyset$, by~\eqref{eq:ind0} we have
\begin{eqnarray}
I\!\left(V; \{W_k\}_{k\in[K]}\right) = 0,
\end{eqnarray}
which implies
\begin{eqnarray}
H(V) = H\!\left(V \mid \{W_k\}_{k\in[K]}\right).
\label{eq:lemma2tt1}
\end{eqnarray}
Moreover, conditioning reduces entropy, yielding
\begin{eqnarray}
H(V)
\ge H\!\left(V \mid \{W_k\}_{k\in[K]\setminus\{1\}}\right)
\ge H\!\left(V \mid \{W_k\}_{k\in[K]}\right).
\label{eq:lemma2tt2}
\end{eqnarray}
Combining~\eqref{eq:lemma2tt1} and~\eqref{eq:lemma2tt2}, we obtain
\begin{eqnarray}
H(V)
= H\!\left(V \mid \{W_k\}_{k\in[K]\setminus\{1\}}\right)
= H\!\left(V \mid \{W_k\}_{k\in[K]}\right).
\label{eq:lemma2tt3}
\end{eqnarray}

Next, we show that $ML \le H(V) \le H(\mathcal{Z}) = ML + o(L)$, which establishes the desired result.
Since $V$ is non-degenerate, there exists a qualified edge $\{V,V_i\}\in\mathcal{E}$ such that
$f(\{V,V_i\})=\{W_k\}_{k\in\mathcal{M}}$ with $|\mathcal{M}|=M$.
By the correctness constraint~\eqref{dec},
\begin{eqnarray}
ML
&\overset{(\ref{h1})}{=}&
H\!\left(\{W_k\}_{k\in\mathcal{M}}\right) \notag\\
&=&
I\!\left(V,V_i; \{W_k\}_{k\in\mathcal{M}}\right)
+ \underbrace{H\!\left(\{W_k\}_{k\in\mathcal{M}} \mid V,V_i\right)}_{\overset{(\ref{dec})}{=}0} \notag\\
&=&
I\!\left(V,V_i; \{W_k\}_{k\in\mathcal{M}}\right) \notag\\
&=&I\!\left(V_i; \{W_k\}_{k\in\mathcal{M}}\right)+I\!\left(V; \{W_k\}_{k\in\mathcal{M}} \mid V_i\right)~~\label{eq:lemma1ttt4}\\
&\overset{(\ref{eq:ind0})}{=}&
I\!\left(V; \{W_k\}_{k\in\mathcal{M}} \mid V_i\right)\le H(V),
\label{eq:lemma1tt4}
\end{eqnarray}
where in (\ref{eq:lemma1ttt4}), we use the fact that $V_i \in \mathcal{V}\backslash\mathcal{V}_d$ such that $\mathcal{C}(V_i) = \emptyset$ and from Lemma \ref{lemma:ind}, $V_i$ is independent of the source symbols $\{W_k\}_{k\in \mathcal{M}}$.
Furthermore,
\begin{eqnarray}
&&H(V)\\
&=&H\left(V|\left\{W_k \right\}_{k\in [K]}\right)+\underbrace{I\left(V; \left\{ W_k \right\}_{k\in [K]} \right)}_{ \overset{(\ref{eq:ind0})}{=}0}\\
&=&\underbrace{H\left(V|\left\{W_k \right\}_{k\in [K]},\mathcal{Z}\right)}_{ \overset{(\ref{messagevn})}{=}0}+I\left(V;\mathcal{Z}|\left\{W_k \right\}_{k\in [K]}\right)~~~\\
&\leq &H(\mathcal{Z})\label{eq:lemma1tt5}\\
&\overset{(\ref{h1})}{=}&L_Z=ML+o(L),\label{eq:lemma1tt6}
\end{eqnarray}
 
Combining the above inequalities yields
\begin{eqnarray}
ML \le H(V) \le H(\mathcal{Z}) = ML + o(L).
\label{eq:lemma2tt7}
\end{eqnarray}
Together with~\eqref{eq:lemma2tt3}, this completes the proof of Lemma~\ref{lemma:size}.

\hfill$\QED$

Next, we show that when $R_Z=1/M$, the two coded symbols
associated with a qualified edge must carry asymptotically independent
information.

\begin{lemma}[Message Size of a Qualified Edge]\label{lemma:qualified}
When $R_Z = 1/M$, for any qualified edge $\{V_i, V_j\}$ associated with source
symbols $\{W_k\}_{k\in \mathcal{M}}$, we have
\begin{eqnarray}
H\!\left(V_i, V_j\right) = 2ML + o(L).
\label{eq:qualified}
\end{eqnarray}
\end{lemma}

{\it Proof:}
We first establish an upper bound. By subadditivity of entropy,
\begin{eqnarray}
H\!\left(V_i, V_j\right)
\le H(V_i) + H(V_j)
\overset{(\ref{eq:size})}{=} 2ML + o(L).
\end{eqnarray}

To obtain a matching lower bound, note that by the correctness constraint~\eqref{dec},
the source symbols $\{W_k\}_{k\in\mathcal{M}}$ are functions of $(V_i, V_j)$, and hence
\begin{eqnarray}
H\!\left(V_i, V_j\right)
&\overset{(\ref{dec})}{=}&
H\!\left(V_i, V_j, \{W_k\}_{k\in\mathcal{M}}\right) \notag\\
&=&
H\!\left(\{W_k\}_{k\in\mathcal{M}}\right)
+ H\!\left(V_i, V_j \mid \{W_k\}_{k\in\mathcal{M}}\right) \notag\\
&\overset{(\ref{h1})}{\geq}&
ML + H\!\left(V_i \mid \{W_k\}_{k\in[K]}\right)\\
&\overset{(\ref{eq:size})}{=}& 2ML + o(L).
\end{eqnarray}
Combining the upper and lower bounds completes the proof.

\hfill$\QED$



Next, we show that when $R_Z = 1/M$, all nodes in a non-degenerate subgraph
$\widetilde{G}$ of $G$ must share the same noise, i.e., the noise must align.
This property is formalized in the following lemma.

\begin{lemma}[Noise Alignment for All Nodes] \label{lemma:noise}
When $R_Z = 1/M$, for a non-degenerate graph $\widetilde{G} = (\mathcal{V}, \mathcal{E})$ such that every node $V \in \mathcal{V}$ satisfies $\mathcal{C}(V) = \emptyset$, we have 
\begin{eqnarray}
H\left(\mathcal{V} \mid \left\{W_k\right\}_{k \in [K]} \right) = ML+o(L). \label{eq:noise0}
\end{eqnarray}
\end{lemma}

{\it Proof:} On the one hand, we have
\begin{eqnarray}
    &&H\left(\mathcal{V} \mid \left\{W_k\right\}_{k \in [K]} \right)\notag\\
    &=&\underbrace{H\left(\mathcal{V}|\left\{W_k \right\}_{k\in [K]},\mathcal{Z}\right)}_{ \overset{(\ref{messagevn})}{=}0}+I\left(\mathcal{V};\mathcal{Z} \mid \left\{W_k\right\}_{k \in [K]} \right)~~\\
    &\leq &H\left(\mathcal{Z} \mid \left\{W_k\right\}_{k \in [K]} \right)\\
    &\leq &H(\mathcal{Z})=L_Z=ML+o(L).
\end{eqnarray}

On the other hand, for any $V \in \mathcal{V}$,
\begin{eqnarray}
&&H\left(\mathcal{V} \mid \left\{W_k\right\}_{k \in [K]} \right)\notag\\
&\geq& H\left(V \mid \left\{W_k\right\}_{k \in [K]} \right)\overset{(\ref{eq:size})}{=}ML+o(L).
\end{eqnarray}
The proof is now complete.

\hfill\QED

Consider the nodes $V_1,\ldots,V_P$ that violate the condition in
Theorem~\ref{thm:d2}. Specifically, for each pair
$(V_p,V_{p+1})$, $p\in[P-1]$, no information about $W_1$ can be decoded,
whereas $W_1$ is decodable from $(V_1,V_P)$.
The following lemma shows that the coded symbols
$V_1,\ldots,V_P$ must carry aligned information about $W_1$ and the noise.

\begin{lemma}[Coded Symbol Alignment]\label{lemma:signal}
When $R_Z = 1/M$, for the nodes $V_1,\ldots,V_P$ specified above, we have
\begin{align}
H\!\left(V_p, V_{p+1} \mid \{W_k\}_{k\in[K]\setminus\{1\}}\right)
&= ML + o(L), ~ \forall p\in[P-1], \label{eq:signal}\\
H\!\left(V_1, V_P \mid \{W_k\}_{k\in[K]\setminus\{1\}}\right)
&= ML + o(L). \label{eq:signal0}
\end{align}
\end{lemma}

{\it Proof:}
We first prove~\eqref{eq:signal}.
The lower bound follows immediately from
\begin{eqnarray}
&&H\!\left(V_p, V_{p+1} \mid \{W_k\}_{k\in[K]\setminus\{1\}}\right)\notag\\
&\ge&
H\!\left(V_p \mid \{W_k\}_{k\in[K]\setminus\{1\}}\right)
\overset{(\ref{eq:size})}{=} ML + o(L). \label{eq:lm5s1}
\end{eqnarray}

We next establish the upper bound.
Since $(V_p,V_{p+1})$ does not allow decoding of $W_1$, the edge
$\{V_p,V_{p+1}\}$ is either unqualified, or qualified with
$W_1\notin f(\{V_p,V_{p+1}\})=\{W_k\}_{k\in \mathcal{M}}$.

\emph{Case 1:} $\{V_p,V_{p+1}\}$ is an unqualified edge, and $\mathcal{M=\emptyset}$.
By the security constraint~\eqref{sec},
\begin{eqnarray}
&& H\left(V_{p}, V_{p+1} \mid \left\{W_k\right\}_{k\in[K]\backslash\{1\}} \right) \notag\\
&=&  H\left(V_{p}, V_{p+1} \mid \left\{W_k\right\}_{k\in[K]} \right)\notag\\
&&+I\left(V_{p}, V_{p+1};W_1 \mid \left\{W_k\right\}_{k\in[K]\backslash\{1\}} \right) \\
&=&  H\left(V_{p}, V_{p+1} \mid \left\{W_k\right\}_{k\in[K]} \right)\notag\\
&&+\underbrace{I\left(V_{p}, V_{p+1}; \left\{W_k\right\}_{k\in[K]} \right)}_{ \overset{(\ref{sec})}{=}0} \notag\\
&&-\underbrace{I\left(V_{p}, V_{p+1}; \left\{W_k\right\}_{k\in[K]\backslash\{1\}} \right)}_{ \overset{}{=}0} \label{eq:lm5tttt1}\\
 &\leq&  H\left( \mathcal{V}\mid \left\{W_k\right\}_{k\in[K]} \right) \label{eq:s2}\\
 &\overset{(\ref{eq:noise0})}{=}& ML+o(L),
\end{eqnarray}
where the third term of \eqref{eq:lm5tttt1} is zero since
$0 \le I(V_{p}, V_{p+1}; \{W_k\}_{k\in[K]\setminus\{1\}})
\le I(V_{p}, V_{p+1}; \{W_k\}_{k\in[K]})
\overset{(\ref{sec})}{=} 0$,
which implies
$I(V_{p}, V_{p+1}; \{W_k\}_{k\in[K]\setminus\{1\}})=0$.

\emph{Case 2:} $\{V_p,V_{p+1}\}$ is a qualified edge with
$f(\{V_p,V_{p+1}\})$ $=$ $\{W_k\}_{k\in\mathcal{M}}$ and $W_1\notin\{W_k\}_{k\in\mathcal{M}}$.
By correctness~\eqref{dec} and Lemma~\ref{lemma:qualified},
\begin{align}
&H\left(V_{p}, V_{p+1} \mid \left\{W_k\right\}_{k\in[K]\backslash\{1\}} \right)\notag\\ \overset{(\ref{dec})}{=}& H\left(V_{p}, V_{p+1}\right) - I\left(V_{p}, V_{p+1},\left\{W_k\right\}_{k \in \mathcal{M}} ;\left\{W_k\right\}_{k\in[K]\backslash\{1\}} \right) \\
\overset{(\ref{eq:qualified})}{=}& 2ML+o(L) - I\left(\left\{W_k\right\}_{k \in \mathcal{M}} ;\left\{W_k\right\}_{k\in[K]\backslash\{1\}} \right)\notag\\
&-\underbrace{I\left(V_{p}, V_{p+1};\left\{W_k\right\}_{k\in[K]\backslash(\{1\}\cup\mathcal{M})}|\left\{W_k\right\}_{k \in \mathcal{M}}  \right)}_{ \overset{}{=}0}\label{eq:lem5tt1}\\
\overset{(\ref{h1})}{=}&ML+o(L),
\end{align}
where the second term in \eqref{eq:lem5tt1} equals $ML$ since $W_1\notin $ $\{W_k\}_{k\in\mathcal{M}}$ and $\{W_k\}_{k\in\mathcal{M}}\subseteq\{W_k\}_{k\in[K]\backslash\{1\}}$, which implies $I(\{W_k\}_{k\in\mathcal{M}};\{W_k\}_{k\in[K]\backslash\{1\}})=H(\{W_k\}_{k\in\mathcal{M}})=$ $ML$; the third term in \eqref{eq:lem5tt1} is zero because $0\le I(V_p,V_{p+1};\{W_k\}_{k\in[K]\backslash(\{1\}\cup\mathcal{M})}\mid\{W_k\}_{k\in\mathcal{M}})\le I(V_p,V_{p+1};\{W_k\}_{k\in[K]\backslash\mathcal{M}}\mid\{W_k\}_{k\in\mathcal{M}})\overset{(\ref{sec})}{=}0$, and hence $I(V_p,V_{p+1};\{W_k\}_{k\in[K]\backslash(\{1\}\cup\mathcal{M})}\mid\{W_k\}_{k\in\mathcal{M}})=0$.

Combining the two cases and (\ref{eq:lm5s1}) proves~\eqref{eq:signal}.

We now prove~\eqref{eq:signal0}.
The lower bound follows from
\begin{eqnarray}
&&H\!\left(V_1, V_P \mid \{W_k\}_{k\in[K]\setminus\{1\}}\right)\notag\\
&\ge&
H\!\left(V_1 \mid \{W_k\}_{k\in[K]\setminus\{1\}}\right)
\overset{(\ref{eq:size})}{=} ML + o(L).
\end{eqnarray}

For the upper bound, by submodularity of entropy,
\begin{eqnarray}
&&\sum_{p=1}^{P-1}
H\!\left(V_p, V_{p+1} \mid \{W_k\}_{k\in[K]\setminus\{1\}}\right)\notag\\
&\ge&
H\!\left(V_1,\ldots,V_P \mid \{W_k\}_{k\in[K]\setminus\{1\}}\right) \notag\\
&&+
\sum_{p=2}^{P-1}
H\!\left(V_p \mid \{W_k\}_{k\in[K]\setminus\{1\}}\right).
\end{eqnarray}
Using~\eqref{eq:signal} and~\eqref{eq:size}, we obtain
\begin{eqnarray}
&&(P-1)ML + o(L)
\ge H\!\left(V_1, V_P \mid \{W_k\}_{k\in[K]\setminus\{1\}}\right)\notag\\
&&~~~~~~~~~~~~~~~~~~~~~~~~~~~+ (P-2)ML + o(L)\\
&\Rightarrow&H\!\left(V_1, V_P \mid \{W_k\}_{k\in[K]\setminus\{1\}}\right)\leq ML+o(L).
\end{eqnarray}
This completes the proof.

\hfill$\QED$

After establishing the above lemmas, we are ready to demonstrate the contradiction as follows. Recall that from $(V_1,V_P)$, we can recover $W_1$, i.e., $W_1\in f(\{V_1,V_P\})$. 
\begin{eqnarray}
&&ML+o(L) \notag\\ &\overset{(\ref{eq:signal0})}{=}& H\left(V_1, V_P \mid \left\{W_k\right\}_{k\in[K]\backslash\{1\}} \right) \\
&\overset{(\ref{dec})}{=}& H\left(V_1, V_P, W_1 \mid \left\{W_k\right\}_{k\in[K]\backslash\{1\}} \right) \\
&=& H\left( W_1 \mid \left\{W_k\right\}_{k\in[K]\backslash\{1\}} \right) \notag\\
&&+~H\left(V_1, V_P \mid \left\{W_k\right\}_{k\in[K]} \right) \\
&\geq& H\left( W_1 \mid \left\{W_k\right\}_{k\in[K]\backslash\{1\}} \right) \notag\\
&&+~H\left(V_1 \mid \left\{W_k\right\}_{k\in[K]} \right) \\
&\overset{(\ref{h1})(\ref{eq:size})}{=}&  L + ML+o(L). \label{eq:d1}
\end{eqnarray}
Normalizing \eqref{eq:d1} by $L$ and letting $L\to\infty$, we obtain $M \geq 1+M$, which is a contradiction. Hence, the proof of the only-if part is complete.

\subsection{Proof of the ``If'' Part}\label{sec:thm22}

We prove that under the condition of Theorem~\ref{thm:d2}, the source key capacity equals $1/M$. The proof proceeds in two steps. We first establish the upper bound $R_Z \le 1/M$, and then show that this rate is achievable.

\paragraph*{Upper bound.}
Since $\widetilde{G}$ is nonempty, there exists a qualified edge $\{V_i, V_j\}$ associated with a source set $\{W_k\}_{k \in \mathcal{M}}$, where $|\mathcal{M}| = M$. Moreover, by assumption, for any $V_j \in \mathcal{V}\setminus \mathcal{V}_d$ we have $\mathcal{C}(V_j) = \emptyset$. The correctness constraint~\eqref{dec} implies
\begin{eqnarray}
&&ML\notag\\
&\overset{(\ref{h1})}{=}& H\left(\left\{W_k\right\}_{k\in\mathcal{M}}\right)\\
&\overset{(\ref{dec})}{=}& I\left(V_i, V_j; \left\{W_k\right\}_{k\in\mathcal{M}}\right)\\
&\overset{(\ref{eq:ind0})}{=}& I\left(V_i; \left\{W_k\right\}_{k\in\mathcal{M}} | V_j \right) \label{eq:s1}\\
&\leq& H(V_i) \\
&=&H\left(V_i|\left\{W_k \right\}_{k\in [K]}\right)+\underbrace{I\left(V_i; \left\{ W_k \right\}_{k\in [K]} \right)}_{ \overset{(\ref{eq:ind0})}{=}0}\\
&=&\underbrace{H\left(V_i|\left\{W_k \right\}_{k\in [K]},\mathcal{Z}\right)}_{ \overset{(\ref{messagevn})}{=}0}+I\left(V_i;\mathcal{Z}|\left\{W_k \right\}_{k\in [K]}\right)~~~\\
&\leq &H(\mathcal{Z})\\
&\overset{(\ref{h1})}{=}&L_Z. \label{eq:s1t1}
\end{eqnarray}
The equality in~\eqref{eq:s1} follows from the fact that $\mathcal{C}(V_j)=\emptyset$ for all $V_j \in \mathcal{V}\setminus \mathcal{V}_d$, together with the independence condition~\eqref{eq:ind0}. Therefore,
\begin{eqnarray}
    R \overset{(\ref{rate})}{=} \frac{L}{L_Z} \overset{(\ref{eq:s1t1})}{\le} \frac{1}{M}.
\end{eqnarray}

We now present a secure storage code construction that achieves the source key rate $R_Z = 1/M$ when the graph $G=(\mathcal{V},\mathcal{E})$ satisfies the condition in Theorem~\ref{thm:d2}. The proposed scheme generalizes the construction in Section~\ref{sec:thm11}.

Let the alphabet of each source symbol $W_k$ be the finite field $\mathbb{F}_q$, where $q > M|\mathcal{E}|$. 
The source key $\mathcal{Z}$ consists of $M$ field symbols from $\mathbb{F}_q$. 
Each coded symbol $V_n$ consists of $M$ field symbols from $\mathbb{F}_q$. 
Consequently, the source symbol $W_k$ occupies one field symbol, while the source key $\mathcal{Z}$ occupies $M$ field symbols, yielding the source key rate $R_Z = H(W_k)/H(\mathcal{Z}) = 1/M$, as desired.

Degenerate nodes $\mathcal{V}_d$ together with their incident edges are trivial, and thus it suffices to focus on the non-degenerate subgraph $\widetilde{G}$ of $G$.


\textbf{Characteristic graphs $G^{[k]}$:}  
Fix any $k \in [K]$ and consider the characteristic graph $G^{[k]}$, which consists of $U^{[k]}$ unqualified components. Choose a prime field size $q$ satisfying $q > \max_{k \in [K]} U^{[k]}$.

For each unqualified component indexed by $u \in [U^{[k]}]$, all nodes within this component are assigned the same coded symbol. Specifically, for any node $V^{[k]}$ in the $u$-th unqualified component of $G^{[k]}$, define
\begin{eqnarray}
V^{[k]} = {\bf h}_{u}^{[k]} W_k + {\bf Z}, \label{eq:cc1}
\end{eqnarray}
where ${\bf h}_{u}^{[k]} \in \mathbb{F}_q^{M \times 1}$ is a coding vector and ${\bf Z} \in \mathbb{F}_q^{M \times 1}$ is an i.i.d. uniform noise vector, independent of all source symbols. The same source key ${\bf Z}$ is shared across all characteristic graphs.

The assignment in \eqref{eq:cc1} satisfies the following properties:

\begin{itemize}
    \item \textbf{Recovery over qualified edges:}  
    For any qualified edge $\{V_i^{[k]}, V_j^{[k]}\}$ corresponding to a set of source symbols $\{W_k\}_{k\in\mathcal{M}} \subseteq \{W_k\}_{k\in [K]}$, the two nodes belong to different unqualified components and thus have linearly independent coding vectors ${\bf h}_{u_i}^{[k]}$ and ${\bf h}_{u_j}^{[k]}$. Therefore,
    \begin{eqnarray}
    V_i^{[k]} - V_j^{[k]} = \big({\bf h}_{u_i}^{[k]} - {\bf h}_{u_j}^{[k]}\big) W_k={\bf h}_{ij}^{[k]}  W_k,  k \in \mathcal{M},
    \label{eq:cc22}
    \end{eqnarray}

    \item \textbf{Equality over unqualified edges:}  
    For any unqualified edge $\{V_i^{[k]}, V_j^{[k]}\}$ associated with source symbols $\{W_k\}_{k\in [K] \setminus \mathcal{M}}$, both nodes lie in the same unqualified component. Hence,
    \begin{eqnarray}
    V_i^{[k]} &=& V_j^{[k]}, \quad k \in [K] \setminus \mathcal{M}.
    \label{eq:cc333}
    \end{eqnarray}
\end{itemize}

\textbf{Global encoding.}
For the original graph $G$, the coded symbol stored at each node $V$ is defined as the superposition
\begin{eqnarray}
V = \sum_{k\in[K]} V^{[k]}. \label{eq:cc2}
\end{eqnarray}

We now show that there exists a choice of coding vectors ${\bf h}_{u}^{[k]}$, $k\in[K]$, $u\in[U^{[k]}]$, such that the construction in \eqref{eq:cc1} and \eqref{eq:cc2} satisfies both correctness and security. To this end, each entry of ${\bf h}_{u}^{[k]}$ is chosen independently and uniformly at random from $\mathbb{F}_q$.

The source key rate is given by $R_Z=L/L_Z=1/M$, as desired.

\textbf{Correctness.}
Consider any qualified edge $\{V_i,V_j\}$ with $f(\{V_i,V_j\})=\mathcal{M}$ and $|\mathcal{M}|=M$. We have
\begin{eqnarray}
V_i - V_j 
&\overset{(\ref{eq:cc2})}{=}& \sum_{k\in[K]} \left(V^{[k]}_i - V^{[k]}_j\right) \notag\\
&\overset{(\ref{eq:cc1})(\ref{eq:cc333})}{=}& \sum_{k\in\mathcal{M}} \left(V^{[k]}_i - V^{[k]}_j\right) \label{eq:t1}\\
&\overset{(\ref{eq:cc1})(\ref{eq:cc22})}{=}& {\bf H}_{ij} \left\{W_k\right\}_{k\in\mathcal{M}}. \label{eq:t2}
\end{eqnarray}
Here, equation~\eqref{eq:t1} follows because for any $k \in [K] \setminus \mathcal{M}$, the edge $\{V^{[k]}_i, V^{[k]}_j\}$ is unqualified. By construction, both nodes lie in the same unqualified component of $G^{[k]}$, which implies $V^{[k]}_i = V^{[k]}_j$ as shown in \eqref{eq:cc333}.  
Equation~\eqref{eq:t2} holds because for each $k \in \mathcal{M}$, the edge $\{V^{[k]}_i, V^{[k]}_j\}$ is qualified and spans two distinct unqualified components in $G^{[k]}$, since no qualified edge is internal to an unqualified component. Consequently, the difference $V^{[k]}_i - V^{[k]}_j$ contributes a linear term in $W_k$, as shown in \eqref{eq:cc22}. Collecting these terms for all $k \in \mathcal{M}$ yields an $M \times M$ linear transformation over $\mathbb{F}_q$, characterized by the matrix ${\bf H}_{ij}$ whose entries are determined by the coding vectors ${\bf h}_{u}^{[k]}$.

Correct decoding requires that each matrix ${\bf H}_{ij}$ be full rank. 
Consider the determinant $|{\bf H}_{ij}|$ as a polynomial in the coding vector variables ${\bf h}_{u}^{[k]}$, for $k \in [K]$ and $u \in [U^{[k]}]$. 
This polynomial has degree $M$ and is not identically zero, since there exists at least one assignment of the coding vectors for which $|{\bf H}_{ij}| \neq 0$.

Define the polynomial
\begin{eqnarray}
\text{poly} \triangleq \prod_{\{V_i, V_j\} \in \mathcal{E}} |{\bf H}_{ij}|,
\end{eqnarray}
which is a nonzero polynomial of degree at most $M |\mathcal{E}|$.  
By the Schwartz–Zippel lemma, if the field size satisfies $q > M |\mathcal{E}|$, there exists a choice of coding vectors such that $\text{poly} \neq 0$.  

Consequently, all matrices ${\bf H}_{ij}$ corresponding to qualified edges are invertible, and the $M$ desired source symbols can be recovered, establishing the correctness of the scheme.

\textbf{Security.}  
We show that each qualified and unqualified edge satisfies the information-theoretic security constraint.  

For any qualified edge $\{V_i, V_j\}$ associated with a set of source symbols $\{W_k\}_{k \in \mathcal{M}}$, we have
\begin{align}
    &I\Big(V_i, V_j; \{W_k\}_{k \in [K]\backslash \mathcal{M}} \,\big|\, \{W_k\}_{k \in \mathcal{M}} \Big) \notag\\
    =&\, I\Big(V_i - V_j, V_j; \{W_k\}_{k \in [K]\backslash \mathcal{M}} \,\big|\, \{W_k\}_{k \in \mathcal{M}} \Big)\\
    \overset{\substack{(\ref{eq:cc1})\\(\ref{eq:cc2})\\(\ref{eq:t2})}}{=}&\, I\Big({\bf H}_{ij} \{W_k\}_{k\in\mathcal{M}}, \sum_{k\in[K]} ({\bf h}_{u}^{[k]} W_k + {\bf Z});\notag\\
    &\{W_k\}_{k \in [K]\backslash\mathcal{M}} \,\big|\, \{W_k\}_{k \in \mathcal{M}} \Big)\\
    =&\, H\Big({\bf H}_{ij} \{W_k\}_{k\in\mathcal{M}}, \sum_{k\in[K]} ({\bf h}_{u}^{[k]} W_k + {\bf Z}) \,\big|\, \{W_k\}_{k \in \mathcal{M}} \Big) \notag\\
    & -H\Big({\bf H}_{ij} \{W_k\}_{k\in\mathcal{M}}, \sum_{k\in[K]} ({\bf h}_{u}^{[k]} W_k + {\bf Z}) \,\big|\, \{W_k\}_{k \in [K]} \Big)\\
    =&\, H\Big(\sum_{k\in[K]\setminus \mathcal{M}} {\bf h}_{u}^{[k]} W_k + K{\bf Z} \,\big|\, \{W_k\}_{k \in \mathcal{M}} \Big) \notag\\
    &- H\Big(K {\bf Z} \,\big|\, \{W_k\}_{k \in [K]} \Big)\\
    =&\, H\Big(\sum_{k\in[K]\setminus \mathcal{M}} {\bf h}_{u}^{[k]} W_k + K{\bf Z}, \{W_k\}_{k \in \mathcal{M}} \Big)\notag\\
       & - H\Big(\{W_k\}_{k \in \mathcal{M}}\Big) - H\Big(K {\bf Z} \,\big|\, \{W_k\}_{k \in [K]} \Big) \label{secutt1}\\
    \overset{(\ref{ind})(\ref{h1})}{=}&\, 2ML - ML - ML = 0,
\end{align}
where the third term in \eqref{secutt1} equals $ML$ because the source key ${\bf Z}$ is independent of the source symbols $\{W_k\}_{k \in [K]}$ and ${\bf Z} \in \mathbb{F}_q^{M \times 1}$.  

For any unqualified edge $\{V_i, V_j\}$, where $\mathcal{M} = \emptyset$ and $V_i = V_j$, we have
\begin{eqnarray}
    && I\Big(V_i, V_j; \{W_k\}_{k \in [K]} \Big) \notag\\
    &=& I\Big(V_i; \{W_k\}_{k \in [K]} \Big)\\
    &\overset{(\ref{eq:cc1})(\ref{eq:cc2})}{=}& I\Big(\sum_{k\in[K]} ({\bf h}_{u}^{[k]} W_k + {\bf Z}); \{W_k\}_{k \in [K]} \Big)\\
    &=& H\Big(\sum_{k\in[K]} ({\bf h}_{u}^{[k]} W_k + {\bf Z}) \Big) \notag\\
    &&- H\Big(\sum_{k\in[K]} ({\bf h}_{u}^{[k]} W_k + {\bf Z}) \,\big|\, \{W_k\}_{k \in [K]} \Big)\\
    &=& H\Big(\sum_{k\in[K]} ({\bf h}_{u}^{[k]} W_k + {\bf Z}) \Big) \notag\\
    &&- H\Big(K {\bf Z} \,\big|\, \{W_k\}_{k \in [K]} \Big)\\
    &\overset{(\ref{ind})(\ref{h1})}{=}& ML - ML = 0.
\end{eqnarray}

This completes the proof of security.

\section{Proof of Theorem \ref{thm:2d}} \label{sec:thm3}
In this section, we prove Theorem \ref{thm:2d}. 
We first prove the ``only if'' part in Section \ref{sec:thm31} and then the ``if'' part in Section \ref{sec:thm32}.

\subsection{Proof of the ``Only if'' Part}\label{sec:thm31}

We first state a useful property that holds for any secure storage code constructed without a source key. This property is formalized in the following lemma. 
Note that when $L_Z = o(L)$, we have\footnote{The same proof applies when the $o(L)$ term is exactly zero, i.e., when the rate is exactly achievable.}
\begin{eqnarray}
H\left(V_i, V_j \mid \{ W_k \}_{k \in [K]} \right) \overset{(\ref{messagevn})}{=} o(L). \label{arate}
\end{eqnarray}
Here, $V_i$ and $V_j$ denote the coded symbols stored at the corresponding nodes.

\begin{lemma}[Determinism Given Common Sources]\label{lemma:det}
When $L_Z=o(L)$, for any qualified edge $\{V_i,V_j\}$, the coded symbols $V_i$ and $V_j$ are asymptotically deterministic given their common source symbols, i.e.,
\begin{eqnarray}
H\!\left(
V_i,V_j \,\middle|\,
\{W_k\}_{k\in \mathcal{C}(V_i)\cup \mathcal{C}(V_j)}
\right)
= o(L).
\label{eq:det}
\end{eqnarray}
\end{lemma}

\begin{proof}
Consider any qualified edge $\{V_i,V_j\}$ with
$f(\{V_i,V_j\})=\{W_k\}_{k\in \mathcal{M}}$ and $|\mathcal{M}|=M$.
From the security constraint~\eqref{sec}, we have
\begin{eqnarray}
&&H\!\left(V_i,V_j \mid \{W_k\}_{k\in\mathcal{M}}\right)\notag\\
&=& H\!\left(V_i,V_j \mid \{W_k\}_{k\in[K]}\right) \notag\\
&&+ I\!\left(
V_i,V_j;
\{W_k\}_{k\in[K]\setminus\mathcal{M}}
\,\middle|\,
\{W_k\}_{k\in\mathcal{M}}
\right) \\
&\overset{(\ref{arate}),(\ref{sec})}{=}& o(L).
\label{eq:cvt1}
\end{eqnarray}

From~\eqref{eq:cvt1}, we further obtain
\begin{eqnarray}
o(L)&=& H\left( V_i \mid \left\{W_k\right\}_{k \in \mathcal{M}} \right) \\
&=& H\left( V_i \mid \left\{W_k\right\}_{k \in \mathcal{C}(V_i)}, \left\{W_k\right\}_{k \in \mathcal{M}\backslash\mathcal{C}(V_i)} \right) \\
&=& H\left( V_i \mid \left\{W_k\right\}_{k \in \mathcal{C}(V_i)} \right) -  \notag\\
&&I\left(V_i;\left\{W_k\right\}_{k \in \mathcal{M}\backslash\mathcal{C}(V_i)}\mid \left\{W_k\right\}_{k \in \mathcal{C}(V_i)}    \right) \\
&\geq& H\left( V_i \mid \left\{W_k\right\}_{k \in \mathcal{C}(V_i)} \right) - \notag\\
&&I\left(V_i;\left\{W_k\right\}_{k \in [K]\backslash\mathcal{C}(V_i)} \mid \left\{W_k\right\}_{k \in \mathcal{C}(V_i)} \right) ~~\\
&\overset{(\ref{eq:ind})}{=}& H\left( V_i \mid \left\{W_k\right\}_{k \in \mathcal{C}(V_i)} \right). \label{eq:cvt2}
\end{eqnarray}

Replacing $i$ with $j$ gives
\begin{eqnarray}
H\!\left(
V_j \mid \{W_k\}_{k\in\mathcal{C}(V_j)}
\right)
= o(L).
\label{eq:cvt3}
\end{eqnarray}

Finally, by the subadditivity of conditional entropy, we have
\begin{eqnarray}
    &&H\left(V_i,V_j \mid \left(W_k \right)_{k \in \mathcal{C}(V_i)\cup\mathcal{C}(V_j)} \right)\notag\\
    &\leq&H\left(V_i\mid \left(W_k \right)_{k \in \mathcal{C}(V_i)\cup\mathcal{C}(V_j)} \right)\notag\\
    &&+H\left(V_j \mid \left(W_k \right)_{k \in \mathcal{C}(V_i)\cup\mathcal{C}(V_j)} \right)\\
    &\leq&H\left(V_i\mid \left(W_k \right)_{k \in \mathcal{C}(V_i)} \right)\notag\\
    &&+H\left(V_j \mid \left(W_k \right)_{k \in \mathcal{C}(V_j)} \right)\\
    &\overset{(\ref{eq:cvt2})(\ref{eq:cvt3})}{=}&o(L).
\end{eqnarray}

Since conditional entropy is non-negative, the proof is complete.
\end{proof}

\hfill \QED

With Lemma~\ref{lemma:det} in hand, we are ready to prove the ``only if'' part of Theorem~\ref{thm:2d}.  
We show that if the condition in Theorem~\ref{thm:2d} is violated, then a secure storage code without a source key cannot exist.  
The proof is by contradiction: suppose that there exists a secure storage code with $L_Z=0$.

Suppose the condition is violated; that is, there exists a qualified edge $\{V_i, V_j\}$ such that $\mathcal{C}(V_i) \cup \mathcal{C}(V_j) \neq f(\{V_i, V_j\}) = \{W_k\}_{k \in \mathcal{M}}$.  
Since $\mathcal{C}(V_i) \cup \mathcal{C}(V_j) \subseteq \{W_k\}_{k \in \mathcal{M}}$ and the union is not equal to $\{W_k\}_{k \in \mathcal{M}}$, it follows that $|\mathcal{C}(V_i) \cup \mathcal{C}(V_j)| < M$.

Consider the joint entropy $H(V_i,V_j)$. Using the correctness constraint~\eqref{dec}, we have
\begin{align}
H(V_i,V_j)
&\overset{(\ref{dec})}{=}
H\!\left(V_i,V_j,\{W_k\}_{k\in\mathcal{M}}\right) \notag\\
&=
H\!\left(\{W_k\}_{k\in\mathcal{M}}\right)
+
H\!\left(V_i,V_j \mid \{W_k\}_{k\in\mathcal{M}}\right) \notag\\
&\overset{(\ref{ind})}{\ge}
ML + o(L).
\label{eq:pf3t1}
\end{align}

On the other hand, we can also bound $H(V_i,V_j)$ using Lemma~\ref{lemma:det}:
\begin{align}
H(V_i,V_j)
&=
H\!\left(
V_i,V_j \mid \{W_k\}_{k\in \mathcal{C}(V_i)\cup\mathcal{C}(V_j)}
\right) \notag\\
&\quad
+ I\!\left(
V_i,V_j;
\{W_k\}_{k\in \mathcal{C}(V_i)\cup\mathcal{C}(V_j)}
\right) \notag\\
&\overset{(\ref{eq:det})}{\le}
o(L)
+
H\!\left(
\{W_k\}_{k\in \mathcal{C}(V_i)\cup\mathcal{C}(V_j)}
\right) \notag\\
&\overset{(\ref{h1})}{=}
|\mathcal{C}(V_i)\cup\mathcal{C}(V_j)|\,L + o(L) \notag\\
&<
ML + o(L).
\end{align}

Combining the two bounds yields
$ML+o(L) \le H(V_i,V_j) < ML+o(L)$,
which is a contradiction.  
Therefore, a secure storage code with $L_Z=0$ cannot exist, completing the proof of the ``only if'' part.

\subsection{Proof of the ``If'' Part}\label{sec:thm32}

We now show that if the condition in Theorem~\ref{thm:2d} is satisfied, then a secure storage code without a source key exists, i.e., $L_Z=0$.

{\bf Code Construction:}  
Consider any graph $G=(\mathcal{V},\mathcal{E})$ that satisfies the condition in Theorem~\ref{thm:2d}.  
Let the alphabet of each source symbol $W_k$ be the finite field $\mathbb{F}_q$, where $q > M|\mathcal{E}|$.  
Each coded symbol $V_n$ consists of $|\mathcal{C}(V_n)|$ symbols from $\mathbb{F}_q$. Specifically, we define
\begin{eqnarray}
V_n = {\bf H}_n \times \{W_k\}_{k\in\mathcal{C}(V_n)}, \quad \forall n\in[N],
\label{eq:cc}
\end{eqnarray}
where $\{W_k\}_{k\in\mathcal{C}(V_n)}\in\mathbb{F}_q^{|\mathcal{C}(V_n)| \times 1}$ is a column vector stacking the corresponding source symbols, and ${\bf H}_n\in\mathbb{F}_q^{|\mathcal{C}(V_n)| \times |\mathcal{C}(V_n)|}$ is the associated coding matrix.  

To prove existence, we generate each ${\bf H}_n$, $n\in[N]$, randomly by choosing its entries independently and uniformly from $\mathbb{F}_q$.

{\bf Security Verification:}  
Since the condition in Theorem~\ref{thm:2d} holds, we have the following:  
\begin{itemize}
    \item For any qualified edge $\{V_i,V_j\}$ with $f(\{V_i,V_j\})=\{W_k\}_{k\in\mathcal{M}}$ and $|\mathcal{M}|=M$, the union of the associated source sets satisfies $\mathcal{C}(V_i)\cup\mathcal{C}(V_j)=\{W_k\}_{k\in\mathcal{M}}$.  
    \item For any unqualified edge $\{V_i,V_j\}$ with $f(\{V_i,V_j\})=\emptyset$, we have $\mathcal{C}(V_i)\cup\mathcal{C}(V_j)=\emptyset$.
\end{itemize}

From the construction in~\eqref{eq:cc}, the coded symbols $(V_i,V_j)$ for a qualified edge include only the desired source symbols $\{W_k\}_{k\in\mathcal{M}}$ and no undesired source symbols $\{W_k\}_{k\in[K]\setminus \mathcal{M}}$. Since the source symbols are independent, no information about undesired sources is revealed.  
Similarly, for an unqualified edge, we have $\mathcal{C}(V_i)=\mathcal{C}(V_j)=\emptyset$, so no source symbols are assigned, and security is trivially satisfied.

{\bf Correctness Verification:}  
For any qualified edge $\{V_i,V_j\}$, the coded symbols $(V_i,V_j)$ provide $M$ linear combinations of the $M$ desired source symbols. Stacking $V_i$ and $V_j$ together yields
\begin{eqnarray}
[V_i;V_j] = {\bf H}_{ij} \times \{W_k\}_{k\in\mathcal{M}},
\end{eqnarray}
where ${\bf H}_{ij}\in\mathbb{F}_q^{M\times M}$ is determined by ${\bf H}_i$, ${\bf H}_j$, $\mathcal{C}(V_i)$, and $\mathcal{C}(V_j)$.  

To guarantee decodability, we consider the determinant $|{\bf H}_{ij}|$ as a polynomial in the entries of $\{{\bf H}_n\}_{n\in[N]}$. This polynomial has degree $M$ and is not identically zero, since there exists at least one realization of $\{{\bf H}_n\}_{n\in[N]}$ such that $|{\bf H}_{ij}|\neq 0$.  

Define
\begin{eqnarray}
\mathrm{poly} \triangleq \prod_{\{V_i,V_j\}\in\mathcal{E}} |{\bf H}_{ij}|,
\end{eqnarray}
which is a nonzero polynomial of degree at most $M|\mathcal{E}|$.  
By the Schwartz--Zippel lemma~\cite{Demillo_Lipton,Schwartz,Zippel}, a uniform random choice of $\{{\bf H}_n\}_{n\in[N]}$ over $\mathbb{F}_q$ with $q>M|\mathcal{E}|$ ensures that $\mathrm{poly}\neq 0$ with positive probability.  
Hence, there exists a realization of $\{{\bf H}_n\}_{n\in[N]}$ such that $|{\bf H}_{ij}|\neq 0$ for all qualified edges. Consequently, all desired source symbols can be recovered, establishing correctness.

\section{Discussion}
This work studies secure storage over graphs with information theoretic security constraints and characterizes the extremal behavior of the source key capacity. By focusing on graphs that achieve the largest possible source key rates under nontrivial security requirements, we identify precise structural conditions under which such extremal rates are achievable.

Our results show that the source key capacity is fundamentally governed by the interaction between graph topology and alignment constraints imposed by security and correctness. In particular, the presence of internal qualified edges captures an inherent conflict between reusing common noise for secrecy and diversifying coded symbols for decodability, and serves as the key obstruction to achieving extremal source key rates. We further provide a complete characterization of graphs that admit secure storage without using any source key.

While this work focuses on symmetric settings with uniform source sizes and edge requirements, several natural extensions remain open. These include going beyond extremal source key rate formulations, allowing heterogeneous edge constraints, and relaxing the structural assumptions on common sources. Exploring these directions may reveal richer tradeoffs and provide deeper insights into secure storage over general graphs.

\let\url\nolinkurl
\bibliography{Thesis}

@ARTICLE{Ahlswede_Cai_etal,
        AUTHOR = "R. Ahlswede and N. Cai and S.-Y. R. Li and R. W. Yeung",
        TITLE =  "Network information flow",
        JOURNAL = "IEEE Trans. Inform. Theory",
        MONTH = "Jul.",
        YEAR = "2000",
        VOLUME = "46",
        NUMBER = "4",
        PAGES = "1204-1216"}

@ARTICLE{Yossef_Birk_Jayram_Kol_Trans,
  TITLE = {{Index Coding With Side Information}},
  AUTHOR = {{Z. Bar-Yossef and Y. Birk and T. S. Jayram and T. Kol}},
  JOURNAL = {IEEE Trans. on Information Theory},
  YEAR    = "2011",
  MONTH  = "March",
  VOLUME  = "57",
  NUMBER = "3",
  PAGES = " 1479 - 1494" 
  }

@article{Maddah_Ali_Niesen,
  author    = {M. Maddah-Ali and U. Niesen},
  title     = {Fundamental Limits of Caching},
  JOURNAL = {IEEE Trans. on Information Theory},
  YEAR    = "2014",
  MONTH  = "May",
  VOLUME  = "60",
  NUMBER = "5",
  PAGES = " 2856-2867" 
}

@inproceedings{SymPIR,
  title={Protecting data privacy in private information retrieval schemes},
  author={Gertner, Yael and Ishai, Yuval and Kushilevitz, Eyal and Malkin, Tal},
  booktitle={Proceedings of the thirtieth annual ACM symposium on Theory of computing},
  pages={151--160},
  year={1998},
  organization={ACM}
}

@inproceedings{Applebaum_Arkis_Raykov_Vasudevan,
  title={Conditional disclosure of secrets: Amplification, closure, amortization, lower-bounds, and separations},
  author={Applebaum, Benny and Arkis, Barak and Raykov, Pavel and Vasudevan, Prashant Nalini},
  booktitle={Annual International Cryptology Conference},
  pages={727--757},
  year={2017},
  organization={Springer}
}

@inproceedings{Beimel_Survey,
  title={Secret-sharing schemes: a survey},
  author={Beimel, Amos},
  booktitle={International Conference on Coding and Cryptology},
  pages={11--46},
  year={2011},
  organization={Springer}
}

@ARTICLE{Li_Sun_CDS,
  author={Li, Zhou and Sun, Hua},
  journal={IEEE Transactions on Communications}, 
  title={Conditional Disclosure of Secrets: A Noise and Signal Alignment Approach}, 
  year={2022},
  volume={70},
  number={6},
  pages={4052-4062},
  doi={10.1109/TCOMM.2022.3168282}}

@INPROCEEDINGS{Li_Sun_linearCDS,
  author={Li, Zhou and Sun, Hua},
  booktitle={2021 IEEE International Symposium on Information Theory (ISIT)}, 
  title={On the Linear Capacity of Conditional Disclosure of Secrets}, 
  year={2021},
  volume={},
  number={},
  pages={3202-3207},
  doi={10.1109/ISIT45174.2021.9517867}}

@inproceedings{secure_nc,
  title={Secure network coding},
  author={Cai, Ning and Yeung, Raymond W},
  booktitle={Proceedings IEEE International Symposium on Information Theory,},
  pages={323},
  year={2002},
  organization={IEEE}
}

@inproceedings{combination_network,
  title={Network coding gain of combination networks},
  author={Ngai, Chi Kin and Yeung, Raymond W},
  booktitle={Information Theory Workshop},
  pages={283--287},
  year={2004},
  organization={IEEE}
}

@article{bidokhti2016capacity,
  title={Capacity results for multicasting nested message sets over combination networks},
  author={Bidokhti, Shirin Saeedi and Prabhakaran, Vinod M and Diggavi, Suhas N},
  journal={IEEE Transactions on Information Theory},
  volume={62},
  number={9},
  pages={4968--4992},
  year={2016},
  publisher={IEEE}
}

@inproceedings{Sahraei_Gastpar,
  title={{GDSP: A graphical perspective on the distributed storage systems}},
  author={Sahraei, Saeid and Gastpar, Michael},
  booktitle={2017 IEEE International Symposium on Information Theory (ISIT)},
  pages={2218--2222},
  year={2017},
  organization={IEEE}
}

@article{Schwartz,
  title={Fast probabilistic algorithms for verification of polynomial identities},
  author={Schwartz, Jacob T},
  journal={Journal of the ACM (JACM)},
  volume={27},
  number={4},
  pages={701--717},
  year={1980},
  publisher={ACM}
}

@inproceedings{Zippel,
  title={Probabilistic algorithms for sparse polynomials},
  author={Zippel, Richard},
  booktitle={International symposium on symbolic and algebraic manipulation},
  pages={216--226},
  year={1979},
  organization={Springer}
}

@article{Demillo_Lipton,
  title={A probabilistic remark on algebraic program testing},
  author={Demillo, Richard A and Lipton, Richard J},
  journal={Information Processing Letters},
  volume={7},
  number={4},
  pages={193--195},
  year={1978},
  publisher={Elsevier}
}

@article{Cai_Chan,
  title={Theory of secure network coding},
  author={Cai, Ning and Chan, Terence},
  journal={Proceedings of the IEEE},
  volume={99},
  number={3},
  pages={421--437},
  year={2011},
  publisher={IEEE}
}

@article{Maheshwar_Li_Li,
  title={Bounding the coding advantage of combination network coding in undirected networks},
  author={Maheshwar, Shreya and Li, Zongpeng and Li, Baochun},
  journal={IEEE Transactions on Information Theory},
  volume={58},
  number={2},
  pages={570--584},
  year={2012},
  publisher={IEEE}
}

@article{Salimi_Liu_Cui,
  title={Generalized cut-set bounds for broadcast networks},
  author={Salimi, Amir and Liu, Tie and Cui, Shuguang},
  journal={IEEE Transactions on Information Theory},
  volume={61},
  number={6},
  pages={2983--2996},
  year={2015},
  publisher={IEEE}
}

@INPROCEEDINGS{Wang_Ulukus_CDMS, 
author={Wang, Zhusheng and Ulukus, Sennur},  booktitle={2022 IEEE International Symposium on Information Theory (ISIT)},   title={Communication Cost of Two-Database Symmetric Private Information Retrieval: A Conditional Disclosure of Multiple Secrets Perspective},   year={2022},  volume={},  number={},  pages={402-407},  doi={10.1109/ISIT50566.2022.9834819}}

@inproceedings{Sun_Shieh_Graph,
  title={Secret sharing in graph-based prohibited structures},
  author={Sun, Hung-Min and Shieh, Shiuh-Pyng},
  booktitle={Proceedings of INFOCOM'97},
  volume={2},
  pages={718--724},
  year={1997},
  organization={IEEE}
}

@article{wu2023capacity,
  title={The Capacity Region of Distributed Multi-User Secret Sharing under The Perfect Privacy Condition},
  author={Wu, Jiahong and Liu, Nan and Kang, Wei},
  journal={arXiv preprint arXiv:2302.03920},
  year={2023}
}

@article{khalesi2021capacity,
  title={The capacity region of distributed multi-user secret sharing},
  author={Khalesi, Ali and Mirmohseni, Mahtab and Maddah-Ali, Mohammad Ali},
  journal={IEEE Journal on Selected Areas in Information Theory},
  volume={2},
  number={3},
  pages={1057--1071},
  year={2021},
  publisher={IEEE}
}

@INPROCEEDINGS{Li_Zhang_CDSnoise,
  author={Li, Zhou and Qin, Siyan and Zhang, Xiang and Fan, Jihao and Chen, Haiqiang and Caire, Giuseppe},
  booktitle={2025 IEEE International Symposium on Information Theory (ISIT)}, 
  title={Noise Capacity of Conditional Disclosure of Secrets: A Graph-Theoretic Perspective}, 
  year={2025},
  volume={},
  number={},
  pages={01-06},
  doi={10.1109/ISIT63088.2025.11195296}}

@article{Li_Zhang_CDSnoisejournal,
  title={Graph-Theoretic Characterization of Noise Capacity of Conditional Disclosure of Secrets},
  author={Li, Zhou and Qin, Siyan and Zhang, Xiang and Fan, Jihao and Chen, Haiqiang and Caire, Giuseppe},
  journal={arXiv preprint arxiv.org/abs/2510.22671},
  year={2025}
}

@ARTICLE{Li_Sun_securestorage,
  author={Li, Zhou and Sun, Hua},
  journal={IEEE Transactions on Information Forensics and Security}, 
  title={On Extremal Rates of Secure Storage Over Graphs}, 
  year={2023},
  volume={18},
  number={},
  pages={4721-4731},
  doi={10.1109/TIFS.2023.3299183}}

@ARTICLE{Li_Sun_storageovergraph,
  author={Li, Zhou and Sun, Hua},
  journal={IEEE Transactions on Information Theory}, 
  title={On Extremal Rates of Storage Over Graphs}, 
  year={2024},
  volume={70},
  number={4},
  pages={2464-2478},
  doi={10.1109/TIT.2023.3328434}}
\end{document}